\pdfoutput=1
\documentclass[11pt]{article}
\usepackage[margin=1.1in]{geometry}
\usepackage{amsmath,amssymb,amsthm}
\usepackage{booktabs}
\usepackage{graphicx}
\usepackage{enumitem}
\usepackage{cite}
\usepackage{refcount}
\usepackage[title,titletoc]{appendix}
\usepackage{placeins}
\usepackage{needspace}
\usepackage[colorlinks=true]{hyperref}
\usepackage{xcolor}
\definecolor{linkblue}{rgb}{0.10,0.20,0.50}
\hypersetup{linkcolor=linkblue,citecolor=linkblue,urlcolor=linkblue}

\newcommand{\TV}{\ensuremath{\mathrm{TV}}}
\newcommand{\Gmin}{\ensuremath{\Gamma_{\min}}}

\newcommand{\ceil}[1]{\lceil #1 \rceil}
\newcommand{\dd}{\,d}

\theoremstyle{plain}
\newtheorem{mainthm}{Theorem}

\newtheorem{theorem}{Theorem}[section]
\newtheorem{lemma}[theorem]{Lemma}
\newtheorem{proposition}[theorem]{Proposition}
\newtheorem{corollary}[theorem]{Corollary}
\theoremstyle{definition}
\newtheorem{definition}[theorem]{Definition}

\theoremstyle{remark}
\newtheorem{conjecture}[theorem]{Conjecture}
\newtheorem{remark}[theorem]{Remark}

\title{The exact LCU sampling overhead of collective diagonal unitaries:\\
resonances and a continued-fraction dichotomy}
\author{Zhihui Wang$^{1,3}$, Sujit Roy$^{1,2}$, Manil Maskey$^{1}$, Rahul Ramachandran$^{1}$\\[0.3em]
\small $^{1}$IMPACT AI, Office of Data Science and Informatics (ODSI)/NASA MSFC, Huntsville, AL, USA\\
\small $^{2}$The University of Alabama in Huntsville, Huntsville, AL, USA\\
\small $^{3}$Research Institute for Advanced Computer Science (RIACS), Universities Space Research\\
\small Association (USRA), USA\\[0.2em]
\small \texttt{zhihui.wang@nasa.gov}, \texttt{sujit.roy@nasa.gov}}
\date{}
\newcommand{\arxividvone}{2610.03595v1}

\begin{document}
\maketitle
\vspace{-1.2em}
\begin{abstract}
\noindent The minimal sampling overhead of the quadratic collective phase $e^{-i\gamma K^2}$ over one layer of single-qubit rotations is decided by the continued fraction of $\gamma/\pi$. Here $K$ is a collective observable with integer spectrum $\{0,\dots,n\}$, such as the permutation-symmetric Hamming weight. This gate is the cardinality-penalty layer of constrained quantum optimization, the one-axis-twisting gate of spin squeezing and the Kerr phase of a bosonic mode. Instead of compiling it into two-qubit gates, we sample it as an LCU over such layers, at overhead $\Gamma$.

Our results give an optimality theory for this overhead. Ancilla-free sampling reproduces the target's outcome probabilities up to the minimal $\Gamma$, and no smaller factor works for every input; for permutation-symmetric $K$, allowing independent per-qubit angles or one layer of arbitrary single-qubit gates gives no further reduction. Rational angles $\pi p/q$ in lowest terms cost exactly $q$ once $n\ge2q-2$, attained uniquely by a uniform $q$-point sampler; the cost is $\Theta(n)$ if and only if $\gamma/\pi$ is badly approximable, and bounded in $n$ if and only if it is rational. A second result shows what the symmetry buys: once it is removed, the squaring phase of Fourier arithmetic on $m$ qubits costs at least $2^{0.557m-O(1)}$ over single-qubit gates, beyond the $2^{m/2}$ of every cut bound, by a new certificate-lifting argument.

For constructions usually treated only through upper bounds, the theory gives the optimum: the Fourier-based LCU of Carrera Vazquez, Egger and Woerner (2026) is, at fixed $n$, optimal up to a constant at typical angles and beaten by $(n+1)/q$ at rational ones, where the resonant sampler is an $n$-independent recipe; the Kerr decomposition of Upreti, Quesada and Chabaud (2026) is the unique optimum over phase shifters at rational parameters without a photon-number cutoff; and it identifies the rational couplings at which one-axis twisting admits optimal cat decompositions.
\end{abstract}

\section{Introduction}\label{sec:intro}

A quantum circuit that cannot implement a target unitary $U$ directly can often still \emph{sample}
it: write $U = \sum_j c_j V_j$ as a linear combination of implementable unitaries (LCU), draw the branch
$V_j$ with probability $|c_j|/\|c\|_1$, and correct in classical post-processing. The price is a
sampling overhead $\Gamma = \|c\|_1^2$ that multiplies the number of circuit repetitions, by $\Gamma$ under
the ancilla-free sampling guarantee and by $\Gamma^2$ under unbiased quasi-probability estimation; this paper prices $\Gamma$ itself. A unitary is an extreme point of the
operator unit ball, so any nontrivial decomposition carries signed weights and $\Gamma > 1$ is forced. This
quasi-probability currency is standard across error mitigation and circuit knitting \cite{tbg2017, ps2024},
and optimality results exist for individual gates \cite{kmb2024, ps2024}. What is missing is an \emph{optimality theory}: for a given target and implementable family, what is the least
possible $\Gamma$, which decompositions attain it, and what feature of the target sets its value? For collective diagonal unitaries the state of the art is the Fourier-based LCU of Carrera Vazquez, Egger and Woerner \cite{anchor}: decompositions by discrete Fourier analysis, the worst-case guarantee $\Gamma \le n+1$ on $n+1$ collective levels, an ancilla-free sampling pipeline, and a $106$-qubit experiment that runs below the bound by an amount that was not known; their outlook lists understanding the method's limitations among the open directions. For the quadratic phase this paper closes that gap, and closing it returns more than matching lower bounds: it returns an arithmetic law.

We answer these questions for the \emph{quadratic collective phase} $U(\gamma) = e^{-i\gamma K^2}$, with $K$ a
collective observable with integer spectrum $\{0, \ldots, n\}$ --- in the qubit reading, the
permutation-symmetric Hamming weight. The gate has three lives. In constrained quantum optimization it is
the cardinality-penalty layer: QUBO penalties $\lambda(\mathbf{1}^\top x - b)^2$ exponentiate to quadratic
phases in the Hamming weight. In spin metrology it is one-axis twisting $e^{-i\chi J_z^2}$ \cite{ku1993}
($J_z = n/2 - K$; its linear and constant terms are free, Theorem~\ref{main:A}).
For a bosonic mode it is the Kerr phase $e^{-i\chi\hat n^2}$, on which recent classical simulators of
bosonic circuits branch \cite{uqc2026, gha2026}. The implementable family $S$ is the same in all three:
uniform rotations $V(\theta) = e^{i\theta K}$, that is, equal-angle single-qubit $RZ$ layers, global
$z$-rotations, or linear phase shifters. It is abelian and, for $n \ge 2$ and $\gamma \notin \pi\mathbb{Z}$,
does not contain $U(\gamma)$, so no circuit over $S$ of any depth equals the target: sampling is the only
route. With two-qubit gates the target is exact at $\Gamma = 1$: $K^2$ expands into a complete graph of
$\Theta(n^2)$ commuting $ZZ$ rotations, which also pays whatever routing the connectivity forces
\cite{anchor}. So the question is an \emph{exchange rate}\label{sec:intro_exchange}: how many shots does
one pay to give up $\Theta(n^2)$ entangling gates and their routing? The rate is exactly the optimal
constant in the ancilla-free sampling guarantee $\tilde p_x \ge p_x/\Gamma$ \cite{chakraborty, barron} for
inputs with weight on every level, and it is not the local-operations cost of circuit knitting, which is
larger. The answer is one sentence. \emph{The cost of $e^{-i\gamma K^2}$ over one layer of single-qubit
rotations is set by the continued fraction of $\gamma/\pi$: it is exactly $q$ at $\gamma = \pi p/q$ in lowest terms once $n \ge 2q-2$, and it grows linearly in $n$ exactly when $\gamma/\pi$ is badly approximable.}

\begin{figure}[tp]\centering
\includegraphics[width=0.95\textwidth]{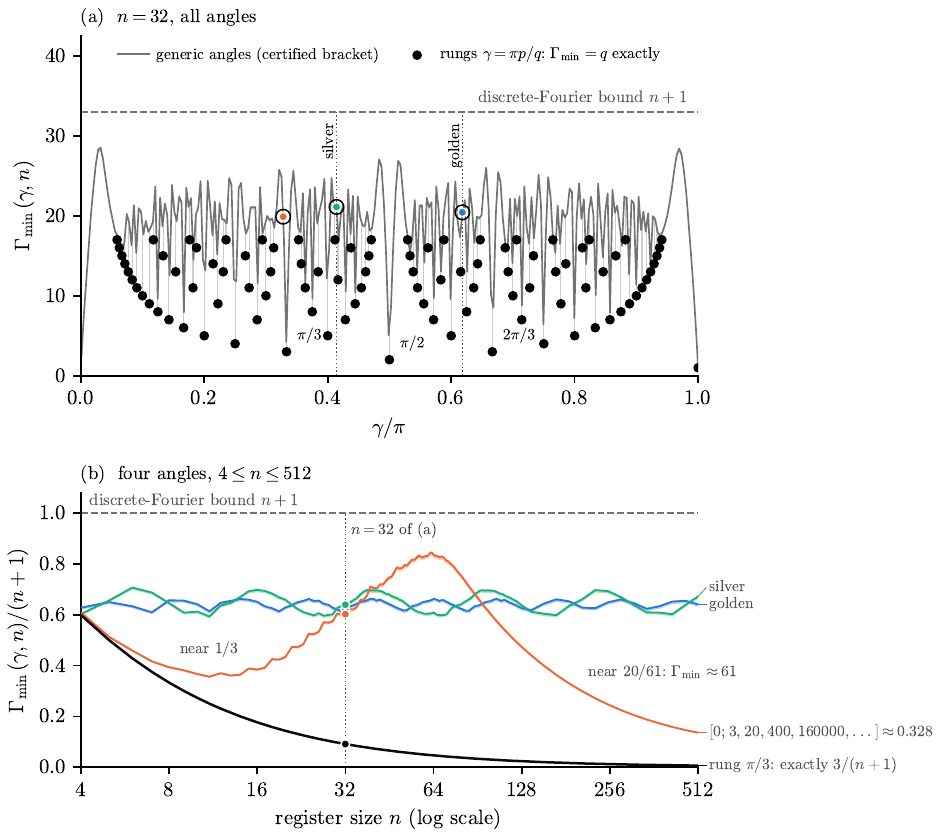}
\caption{The classification, drawn. (a)~The cost $\Gmin(\gamma, n)$ over all angles at $n = 32$; every
value is a certified two-sided bracket of relative width at most $1.6\%$, and the curve is the bracket's upper end, narrower than the line. Dots are the rungs
$\gamma = \pi p/q$, $q \le 17$, where the cost is exactly $q$ (Theorem~\ref{main:B}); between them it
sits on the plateau at about $0.6(n+1)$. (b)~$\Gmin/(n+1)$
against $n$ at four angles, certified brackets: on the rung $\pi/3$ the cost stays at $3$; at the golden
and silver angles $\gamma/\pi = (\sqrt5-1)/2$ and $\sqrt2 - 1$, badly approximable, the ratio holds a constant (Theorem~\ref{main:C}); at
$\gamma/\pi = [0; 3, 20, 400, 160000, \ldots] \approx 0.328$, whose partial quotients grow, the cost follows
each convergent in turn, so the ratio dips far below the plateau infinitely often. The coloured dots in
(a) mark the angles of (b) at $n = 32$.}
\label{fig:landscape}
\end{figure}

The cost is a norm, computed by duality (Theorem~\ref{main:A}):
$\sqrt{\Gmin}$ is the least total variation of a complex measure on the circle whose Fourier coefficients
reproduce the target's spectral data, every lower bound is one trigonometric polynomial of degree $n$ with
sup-norm at most one, and the optimum is a discrete measure with at most $2n+3$ atoms, an implementable
LCU. For the quadratic phase the data are the chirp $e^{-i\gamma k^2}$. At a rational angle $\pi p/q$ the chirp is $q$-periodic after a free modulation,
so $q$ atoms on a coset reproduce every moment at cost $q$, and a Fej\'er-windowed conjugate chirp
certifies that nothing cheaper exists once $n \ge 2q-2$, uniquely (Theorem~\ref{main:B}). The certificate
survives detuning, and the convergents of $\gamma/\pi$ then price every angle: $\Theta(n)$ if and only if
$\gamma/\pi$ is badly approximable, bounded in $n$ if and only if it is rational, and for a bosonic mode
with no photon-number cutoff a finite-cost decomposition over phase shifters exists if and only if it is
rational (Theorem~\ref{main:C}); near a low-denominator rational a master formula, proved two-sided there
and conjectured everywhere, reads the cost off the nearest one (Theorem~\ref{main:D}). Nothing symmetric
helps: independent per-qubit angles and one layer of arbitrary single-qubit gates do not lower the cost of a symmetric
target (Theorem~\ref{main:E}), so every lower bound for a symmetric target is a statement about one layer of single-qubit gates,
and symmetric diagonal Cliffords only renormalize the ladder $2$-adically (Theorem~\ref{main:F}).

Our second result shows what the symmetry buys. Without it, the cost can be exponential: the squaring phase of quantum
Fourier arithmetic on $m$ qubits costs at least $2^{0.557m - O(1)}$ over single-qubit gates, beyond the
$2^{m/2}$ ceiling of every cut bound, by a new certificate-lifting argument (Theorem~\ref{main:register});
this separates implementation resource sets, not complexity classes. Figure~\ref{fig:landscape} draws it:
the rungs at rational angles, the plateau between them, and the dips, as $n$ grows, at angles whose partial quotients are unbounded.

The law prices what is in use. The Fourier-based LCU of Carrera Vazquez, Egger and Woerner \cite{anchor}
guarantees $\Gamma \le n+1$; at each fixed $n$ it is optimal up to a constant factor at typical angles (the
median optimum is provably at least $0.12(n+1)$, and about $0.6(n+1)$ in certified computation at $n \le 128$),
and at rational angles, once $n \ge 2q-2$, the unique optimizer is an $n$-independent recipe, the fractional-revival
decomposition of Kerr optics, at overhead $q$ against $n+1$. The Kerr decomposition of Upreti, Quesada and
Chabaud \cite{uqc2026} is the unique optimum over phase shifters at rational parameters without a photon-number cutoff, so their bound $\sqrt q$ on its $\ell_1$
norm is tight; off the rationals approximation is unavoidable, and recentring the window shortens the
expansion of \cite{gha2026}. With the angle fixed by the physics, the law characterizes one-axis twisting (\S\ref{sec:oat}).

\paragraph{Names.} We call a rational angle $\gamma = \pi p/q$ in lowest terms a \emph{rung} of
denominator $q$, after the ladder the rungs form in Figure~\ref{fig:landscape}(a), and the $\Theta(n)$
regime between the low rungs the \emph{plateau}. A continued fraction is written
$[0; a_1, a_2, \ldots] = 1/(a_1 + 1/(a_2 + \cdots))$, with \emph{partial quotients} $a_j$ and
\emph{convergents} $p_k/q_k$, the fractions obtained by stopping at $a_k$; an irrational number is
\emph{badly approximable} (equivalently, of bounded type) if its partial quotients are bounded. The \emph{threshold} $n_0$ is the least $n$ at which $\Gmin = q$. From \S\ref{sec:resonance} on the data are the quadratic chirp $u_k = e^{-i\gamma k^2}$ and we write $\Gmin(\gamma, n)$.

\paragraph{What is proved, what is measured.} Every statement in
Theorems~\ref{main:A}--\ref{main:register} is proved. What is not proved is labelled where it appears:
\emph{computer-assisted}, over a stated range of $n$; \emph{certified}, a two-sided bracket from an explicit decomposition on a grid of angles above and a dual certificate rescaled by Bernstein's inequality below, at most $1.6\%$ wide for the landscape of Figure~\ref{fig:landscape}; \emph{measured}, as the register exponent
near $0.844$, extrapolated from certified values at $m \le 11$; and \emph{conjectured}, as the master formula
at every angle and the limit of that exponent below $1$.\footnote{Version~1 of this paper, arXiv:\arxividvone, is \cite{v1}, the full record; results stated here without proof are cited to it by statement number.}

\paragraph{Relation to prior work.} What is specific here is that for an \emph{abelian} resource group the optimal-LCU problem collapses to a scalar moment problem on the circle, whose dual certificate is a trigonometric polynomial and then closes on the arithmetic of the angle. The duality itself has two parents: the extent duality of convex resource theories \cite{regula2018, bravyi2019}, whose pure-state identity \cite[Thm.~10]{regula2018} makes the ancilla-free constant exact, and the atomic-norm duality of compressed sensing and super-resolution \cite{crpw2012, tbsr, cf14}, whose dual polynomial is our certificate; the same Hahn--Banach structure underlies the optimal local decompositions of nonlocal gates \cite{ps2024}. The
resonant sampler is the fractional-revival decomposition of Kerr optics
\cite{bb1968, ys1986, mtk1990, tac1993, wg1996}, used for Kerr gates in \cite{uqc2026}, with weights that on
one parity branch are the Zadoff--Chu sequences \cite{fz1962, chu1972}; the converse, optimality and
uniqueness above the threshold, is ours. Theta-sum flatness at bounded partial quotients \cite{hl1914, cf15} sharpens the linear constant, which the dichotomy does not need \cite[\S6.3]{v1}. For the register target the lift goes past the cut bound \cite{nielsen2003, hl2025, sps2025, necaise2026}, with the splitting of a quadratic phase function \cite{bvvdn2014}
and the gluing of the Golay--Rudin--Shapiro recursion \cite{shapiro1951, rudin1959}. The full correspondence between the literatures is \cite[\S11.2]{v1}.

\paragraph{Organization.} Section~\ref{sec:setup} proves that the cost is a norm, computed by duality
(Theorem~\ref{main:A}). Section~\ref{sec:resonance} prices rational angles exactly (Theorem~\ref{main:B}); Section~\ref{sec:generic} proves the continued-fraction dichotomy (Theorem~\ref{main:C}) from two detuning bounds, which also give the master formula on its proved range (Theorem~\ref{main:D});
Section~\ref{sec:enlarge} shows that nothing symmetric helps (Theorems~\ref{main:E} and~\ref{main:F}); Section~\ref{sec:separation} drops the
symmetry (Theorem~\ref{main:register}); Section~\ref{sec:pricing} prices the constructions in use and reads the law in the
metrology direction; Section~\ref{sec:discussion} states the scope and what is open. Appendices~\ref{app:operational}--\ref{app:proofs_separation} hold the proofs not given in the text.

\section{The cost is a norm, computed by duality (Theorem A)}\label{sec:setup}

Throughout, $K = \sum_{k=0}^{n} k\,P_k$ is a collective observable with integer spectrum
$\{0, \ldots, n\}$, $P_k$ its spectral projectors, $N := n+1$ the number of levels, and
$S = \{V(\theta) = e^{i\theta K} : \theta \in [0, 2\pi)\}$ the implementable family, the one-parameter dictionary: for the Hamming
weight, uniform single-qubit $RZ$ layers up to global phase. Targets are diagonal unitaries
$U = \sum_k u_kP_k$ with spectral data $u_k = e^{-i\gamma f(k)}$; from \S\ref{sec:resonance} on,
$f(k) = k^2$ and $u_k = e^{-i\gamma k^2}$. That target class is exactly the reachable set: every
$\int V(\theta)\dd\mu(\theta)$ equals $\sum_k \hat\mu(k) P_k$, which is constant on each eigenspace of
$K$, so a diagonal unitary that is \emph{not} a function of $K$ admits no LCU over $S$ at any finite
cost. For the Hamming weight those are the permutation-symmetric diagonal unitaries; for a weighted
register $K = \sum_j 2^j x_j$ every diagonal unitary on $m$ qubits is a function of $K$, and the same
theorems apply with $N = 2^m$ (\S\ref{sec:separation}). Permutation symmetry is needed only for the enlargement results of \S\ref{sec:enlarge}, not for the theory itself.

\paragraph{Resource model.} Throughout, $\Gmin$ is the overhead of one sampled layer drawn from $S$. For symmetric targets one layer of independent single-qubit angles, or of arbitrary single-qubit gates, costs the same (Theorem~\ref{main:E}); and for this one-layer model the optimal LCU norm, the optimal constant of the ancilla-free sampling guarantee and the optimal channel decomposition over one ancilla-controlled layer all coincide at $\Gmin$ (Propositions~\ref{prop:domination} and~\ref{prop:channel}). Deeper or coherently controlled constructions, two-qubit gates and circuit knitting are separate resource models; where they enter, the text says so.

\begin{definition}[cost norm]\label{def:cost}
\[
\sqrt{\Gmin(U)} \;=\; \inf \Big\{ \|\mu\|_{\TV} \;:\; \mu \text{ a complex measure on the
circle},\ \int V(\theta) \dd\mu(\theta) = U \Big\}.
\]
\end{definition}

Here $\|\mu\|_{\TV} := |\mu|(\mathbb{T})$, $\mathbb{T}$ the circle, is the total variation of the complex measure $\mu$, the total
mass of its modulus; for a discrete measure $\mu = \sum_j c_j\,\delta_{\theta_j}$ it is $\sum_j |c_j|$,
so squaring matches the sampling overhead $\Gamma = \|c\|_1^2$ of a discrete LCU $U = \sum_j c_j V(\theta_j)$. Thus $\sqrt{\Gmin}$ is the norm and $\Gmin$ the overhead: the duality below characterizes the former, and every value of $\Gamma$ in this paper is the square of a total variation, as in $\|\mu^*\|_{\TV} = \sqrt q$ and $\Gmin = q$ at a rung. Since $\{P_k\}$ are linearly independent and
$V(\theta) = \sum_k e^{i\theta k} P_k$, the operator constraint is equivalent to the $N$ moment conditions
\begin{equation}\label{eq:moments}
\hat\mu(k) \;:=\; \int e^{i\theta k} \dd\mu(\theta) \;=\; u_k, \qquad k = 0, \ldots, n :
\end{equation}
$\sqrt{\Gmin}$ is the atomic norm of the data $(u_k)$ with respect to the continuous dictionary of
complex exponentials \cite{crpw2012, tbsr}, and the certificate below is the dual polynomial of
super-resolution theory \cite{cf14}. Exact duality computes it.

\begin{mainthm}[the cost is a norm, computed by duality]\label{main:A}
For every diagonal $U$ that is a function of $K$,
\[
\sqrt{\Gmin(U)} \;=\; \max \Big\{ \Big|\sum_{k=0}^{n} a_k u_k\Big| \;:\;
\Big\|\sum\nolimits_k a_k e^{ik\theta}\Big\|_\infty \le 1 \Big\},
\]
with both optima attained, and the minimizing measure discrete with at most $2n+3$ atoms. Multiplying the data by $e^{i(\alpha k + \beta)}$ does not change the cost.
\end{mainthm}

\begin{proof}
By the Riesz representation theorem, $M(\mathbb{T}) = C(\mathbb{T})^*$. The constraint map
$\mu \mapsto (\hat\mu(k))_{k \le n} \in \mathbb{C}^N$ consists of finitely many
weak-$*$-continuous functionals and is surjective (the discrete-Fourier construction of
\cite{anchor} provides a feasible point for any data). Minimal-norm interpolation duality
(Hahn--Banach, equivalently Fenchel duality against the sup-norm unit ball of polynomials
with spectrum in $\{0, \ldots, n\}$) gives equality. The inf is attained because \TV-balls
are weak-$*$ compact and the constraint set is weak-$*$ closed; the sup by
finite-dimensional compactness.
\end{proof}

\noindent The atom count is Carath\'eodory in $\mathbb{R}^{2N+1}$, applied to the cone over the pairs
(moment vector, \TV): a feasible optimal point is a conic combination of at most $2N+1 = 2n+3$ atomic
generators. Optimal decompositions are therefore explicit, directly implementable atomic LCUs, not
abstract measures; and for any given $(n, \gamma)$ the optimum is a small convex program, since by the
Fej\'er--Riesz theorem in its bounded-real form \cite[Thm.~4.24]{dumitrescu2007}, \cite[Cor.~4.1]{cf14}
the constraint $\|P\|_\infty \le 1$ on the dual is a positive-semidefinite condition on an
$(n+1)\times(n+1)$ Gram matrix (the primal is the Toeplitz program of \cite[Prop.~2.1]{tbsr}). From
here on, duality is the workhorse in both directions: a \emph{lower} bound on the cost is a single
certificate polynomial of degree $n$ and sup-norm one; an \emph{upper} bound is a single feasible
measure. Two consequences are used on every page.

\begin{lemma}[free group actions]\label{lem:free}
For $V \in S$ and unimodular $z$: $\Gmin(zVU) = \Gmin(U)$. Consequently, multiplying the
data by $e^{i(\alpha k + \beta)}$ is free, and for the cardinality penalty
$f(k) = (k-b)^2 = k^2 - 2bk + b^2$ the cost depends only on the chirp coefficient
$\gamma$; the constraint center $b$ is irrelevant.
\end{lemma}

\begin{proof}
Translate and rephase the measure: if $\hat\mu(k) = u_k$ then the measure
$e^{i\beta}\mu(\cdot - \alpha)$ has moments $e^{i(\alpha k + \beta)} u_k$ and the same
total variation. \qedhere
\end{proof}

\noindent This is the last clause of Theorem~\ref{main:A}.

\begin{lemma}[the discrete-Fourier bound \cite{anchor}]\label{lem:dft}
For all unimodular data, $\Gmin \le N = n+1$.
\end{lemma}

\begin{proof}
The $N$-atom equispaced measure whose weights are the inverse DFT of the data reproduces every moment,
and Parseval gives its weight vector $\|c\|_2 = 1$, so $\|c\|_1 \le \sqrt N$ by Cauchy--Schwarz. This is
the construction of \cite{anchor} and its guarantee.
\end{proof}

\paragraph{What $\Gamma$ means.} The ancilla-free guarantee of \cite{chakraborty}, used by \cite{anchor},
samples a rotation from a probability measure and runs it in place of $U$; sampling the LCU then
reproduces every outcome probability of the target up to the factor $\Gamma$. The duality settles that
the constant is exact.

\begin{proposition}[the ancilla-free guarantee is tight; after \cite{regula2018}]\label{prop:domination}
Let $\psi$ be a state with $P_k\psi \ne 0$ for every level $k = 0, \ldots, n$, let $U = \sum_k u_kP_k$,
and for a probability measure $\nu$ on the circle write $\tilde p_x = \int |\langle x|WV(\theta)\psi\rangle|^2
\,\mathrm{d}\nu(\theta)$ and $p_x = |\langle x|WU\psi\rangle|^2$. The least $\Gamma$ for which some $\nu$ satisfies
$p_x \le \Gamma\,\tilde p_x$ for every downstream unitary $W$ and every outcome $x$ is $\Gmin(U)$.
\end{proposition}

\noindent The proof (Appendix~\ref{app:operational}) rewrites the condition as an operator inequality
and applies the pure-state identity between generalized robustness and $\ell_1$-extent of
\cite[Thm.~10]{regula2018}; in our coordinates its rank-one witness is the certificate of
Theorem~\ref{main:A}. The statement is tight only against the \emph{worst} downstream circuit; this is the exact form of the observation in \cite[Sec.~VI\,B]{anchor} that the guarantee is loose in practice yet ``can indeed be tight'' on a small subset of outcomes. The same number is the channel-level cost over one ancilla-controlled layer, and without the ancilla a quasi-probability decomposition exists only at $\Gmin(U) = 1$ (Proposition~\ref{prop:channel}, Appendix~\ref{app:operational}). The identity is a one-layer statement, and that is where the model dependence begins: two ancilla-controlled layers already implement the $q = 2$ rung exactly, at cost one \cite[App.~A]{v1}. Shots therefore grow by $\Gamma$ under the ancilla-free guarantee and by $\Gamma^2$ under unbiased quasi-probability estimation \cite{tbg2017, ps2024}. What $\Gamma$ is not: the cost of circuit knitting, which decomposes the channel into local operations with measurement and classical communication, is larger (\S\ref{sec:separation}, after Conjecture~\ref{conj:exponent}).

\begin{remark}[independent angles]\label{rem:productform}
Nothing in the proof of Proposition~\ref{prop:domination} uses the common angle: over one layer of independent single-qubit diagonal rotations (\S\ref{sec:torus}), for an input with weight on every basis state, the same argument makes the product cost the exact constant of the guarantee for every diagonal $U$, symmetric or not; for the register target of \S\ref{sec:separation} this makes the separation a statement about samplers.
\end{remark}

\begin{proposition}[channel-level cost at one layer]\label{prop:channel}
For a diagonal target $U = \sum_k u_kP_k$, the least $\sum_\alpha|q_\alpha|$ over decompositions
$U\rho U^\dagger = \sum_\alpha q_\alpha\Phi_\alpha(\rho)$ into elementary estimators --- one collective
rotation whose angle is coherently selected by an ancilla, the ancilla measured, the system observable
reweighted by a real number of modulus at most one --- equals $\Gmin(U)$, attained by the standard LCU
circuit on a \TV-optimal atomic measure. Without the ancilla the decomposition exists only if
$\Gmin(U) = 1$.
\end{proposition}

One more consequence of duality is a lower bound that needs no certificate design: with
$P_n(\theta) = \sum_{k=0}^{n} u_k e^{ik\theta}$ the data polynomial, the conjugated data
$a_k = \overline{u_k}/\|P_n\|_\infty$ are admissible and pair to $N/\|P_n\|_\infty$, so
\begin{equation}\label{eq:matched}
\sqrt{\Gmin} \;\ge\; \frac{N}{\|P_n\|_\infty}
\end{equation}
(the \emph{matched filter}). Since the data are unimodular, $\|P_n\|_\infty \ge \|P_n\|_2 = \sqrt N$: no
matched filter proves more than $\Gmin \ge N$ up to a constant, which is where the arithmetic of
\S\S\ref{sec:resonance}--\ref{sec:generic} takes over.

\section{Rational angles cost exactly \texorpdfstring{$q$}{q} (Theorem B)}\label{sec:resonance}

Throughout this section $\gamma = \pi p/q$ with $\gcd(p, q) = 1$, $u_k = e^{-i\gamma k^2}$, and
\[
\sigma \;:=\; pq \bmod 2 \in \{0, 1\}, \qquad
\theta^*_j \;:=\; \frac{2\pi j + \pi\sigma}{q}, \quad j = 0, \ldots, q-1.
\]

\paragraph{The simplest rung, worked out.} At $\gamma = \pi/2$ ($p = 1$, $q = 2$, $\sigma = 0$) the data
$u_k = e^{-i\pi k^2/2} = (1, -i, 1, -i, \ldots)$ are $2$-periodic. Two atoms at $\theta^*_0 = 0$ and
$\theta^*_1 = \pi$ with weights $\tilde c_0 = \frac{1-i}{2}$ and $\tilde c_1 = \frac{1+i}{2}$ reproduce
every moment, $\tilde c_0 + \tilde c_1e^{i\pi k} = u_k$ for all $k$; both weights have modulus
$1/\sqrt2$, so $\|c\|_1 = \sqrt2$ and $\Gamma = 2$. On a coherent state this decomposition is the
Yurke--Stoler cat \cite{ys1986}, $e^{-i(\pi/2)\hat n^2}|\alpha\rangle = \tfrac{1-i}{2}|\alpha\rangle +
\tfrac{1+i}{2}|{-\alpha}\rangle$, phases and all. The degree-two polynomial
$P(\theta) = \frac{1}{2\sqrt2}(1 + 2ie^{i\theta} + e^{2i\theta})$ has
$|P(\theta)| = \sqrt{(1 + \cos^2\theta)/2} \le 1$ and pairs with the data to $\sqrt2$, so no measure does
better: $\Gmin = 2$ for every $n \ge 2$. Everything in this section is this computation with $q$ in
place of $2$.

\begin{lemma}[periodization]\label{lem:period}
$u_{k+q} = (-1)^{pq} u_k$ for all $k \in \mathbb{Z}$. Consequently
$v_k := e^{-i\pi\sigma k/q} u_k$ is $q$-periodic, and the modulation is free by
Lemma~\ref{lem:free}.
\end{lemma}

\begin{proof}
$(k+q)^2 - k^2 = 2kq + q^2$, so $u_{k+q}/u_k = e^{-2\pi i pk} e^{-i\pi pq} = (-1)^{pq}$.
Then $v_{k+q}/v_k = e^{-i\pi\sigma}(-1)^{pq} = (-1)^{\sigma}(-1)^{\sigma} = 1$ by
$\sigma \equiv pq \pmod 2$.
\end{proof}

\begin{proposition}[perfect flatness of the periodized chirp]\label{prop:flat}
Write $v_k = \sum_{j=0}^{q-1} \tilde c_j\, e^{2\pi i jk/q}$ (finite Fourier series of the
$q$-periodic sequence $v$), i.e.
\[
\tilde c_j \;=\; \frac{1}{q} \sum_{k=0}^{q-1} e^{-i\pi(pk^2 + \sigma k)/q}\,
e^{-2\pi i jk/q}.
\]
Then $|\tilde c_j| = 1/\sqrt{q}$ for every $j$.
\end{proposition}

\begin{proof}
Let $S_j = q\tilde c_j$. Since $v$ and the character are $q$-periodic,
\[
|S_j|^2 \;=\; \sum_{h=0}^{q-1} e^{-2\pi i jh/q} \sum_{l=0}^{q-1} v_{l+h}\overline{v_l},
\]
substituting $k = l + h$ over complete residue systems. From
$v_{l+h}\overline{v_l} = e^{-i\pi(ph^2 + \sigma h)/q}\, e^{-2\pi i phl/q}$, the inner sum
is $e^{-i\pi(ph^2+\sigma h)/q} \cdot q\,\mathbf{1}[q \mid ph] = q\,\mathbf{1}[h \equiv 0]$,
because $\gcd(p, q) = 1$. Only $h = 0$ survives: $|S_j|^2 = q$ for every $j$.
\end{proof}

\noindent The periodized chirp is a \emph{perfect polyphase sequence}: its autocorrelation is exactly
impulsive, which is why flatness is exact rather than up to constants, for all $q$. A nontrivial
additive character summing to zero is the only arithmetic input in the paper.

\begin{mainthm}[resonances]\label{main:B}
Let $\gamma = \pi p/q$ with $\gcd(p, q) = 1$.
\begin{enumerate}[label=(\roman*), nosep]
\item \emph{(Resonance.)} The $q$-atom measure $\mu^* = \sum_{j=0}^{q-1} \tilde c_j\, \delta_{\theta^*_j}$
satisfies $\hat\mu^*(k) = u_k$ for \emph{every} $k \in \mathbb{Z}$, and $\|\mu^*\|_{\TV} = \sqrt{q}$.
Hence $\Gmin(\pi p/q,\, n) \le \min(q,\; n+1)$ for every $n \ge 0$. Since the moduli $|\tilde c_j|$ are
constant, the decomposition samples $\theta$ uniformly from the $q$ equally spaced atoms $\theta^*_j$,
with no rejection and no reweighting; all data-dependence sits in the phases.
\item \emph{(Exactness and uniqueness.)} For every $n \ge 2q-2$ and every $p$, $\Gmin(\pi p/q, n) = q$: the degree-$(2q-2)$ polynomial
\[
P(\theta) \;=\; \frac{1}{\sqrt q}\sum_{k=0}^{2q-2}
\Big(1 - \frac{|k - (q-1)|}{q}\Big)\, \overline{u_k}\, e^{ik\theta}
\]
--- the Fej\'er-windowed conjugate chirp --- satisfies $\|P\|_\infty \le 1$ and
$P(\theta^*_j) = \overline{\tilde c_j}/|\tilde c_j|$ for all $j$, so it is an exact dual certificate for
$\mu^*$; and $\mu^*$ is the \emph{unique} \TV-minimizer.
\end{enumerate}
\end{mainthm}

\begin{proof}[Proof of (i)]
By Lemma~\ref{lem:period} and Proposition~\ref{prop:flat}, for all integers $k$,
\[
u_k \;=\; e^{i\pi\sigma k/q} v_k \;=\; \sum_j \tilde c_j e^{i(2\pi j + \pi\sigma)k/q}
\;=\; \sum_j \tilde c_j e^{i\theta^*_j k},
\]
in particular
simultaneously for $k = 0, \ldots, n$, for every $n$, so $n$ plays no role in
feasibility. The total variation is $\sum_j |\tilde c_j| = q \cdot q^{-1/2} = \sqrt{q}$.
The alternative bound $\Gamma \le n+1$ is Lemma~\ref{lem:dft}.
\end{proof}

\begin{proof}[Proof of (ii), exactness]
Write $\mu^* = \sum_j \tilde c_j \delta_{\theta^*_j}$ for the resonant measure of
Theorem~\ref{main:B}(i), so $|\tilde c_j| = 1/\sqrt q$ and
$\overline{u_k} = \sum_j \overline{\tilde c_j}\, e^{-ik\theta^*_j}$. Substituting into the
display of the theorem and collecting the Fej\'er weights,
\[
P(\theta) \;=\; \frac{1}{\sqrt q}\sum_j \overline{\tilde c_j}\,
e^{i(q-1)(\theta - \theta^*_j)}\, F_q(\theta - \theta^*_j),
\qquad
F_q(\varphi) = \sum_{|m| < q}\Bigl(1 - \tfrac{|m|}{q}\Bigr) e^{im\varphi}
= \tfrac{1}{q}\Bigl|\sum_{m=0}^{q-1} e^{im\varphi}\Bigr|^2 \ge 0 .
\]
The atoms form a $q$-point coset, so $\sum_j F_q(\theta - \theta^*_j) = q$ identically,
whence $\|P\|_\infty \le q^{-1/2}\cdot q^{-1/2}\sum_j F_q(\theta - \theta^*_j) = 1$. At
$\theta = \theta^*_i$ every $F_q(\theta^*_i - \theta^*_j)$ with $j \ne i$ vanishes
($\theta^*_i - \theta^*_j \in (2\pi/q)\mathbb{Z}\setminus 2\pi\mathbb{Z}$) and
$F_q(0) = q$, so $P(\theta^*_i) = \sqrt q\, \overline{\tilde c_i}
= \overline{\tilde c_i}/|\tilde c_i|$, the interpolation condition. Finally the pairing of
Theorem~\ref{main:A} is
$\sum_k a_k u_k = q^{-1/2}\sum_{k=0}^{2q-2}\bigl(1 - |k-(q-1)|/q\bigr)|u_k|^2
= q^{-1/2}\sum_{|m|<q}(1 - |m|/q) = q^{-1/2}\cdot q = \sqrt q$. Together with
Theorem~\ref{main:B}(i), $\sqrt{\Gmin} = \sqrt q$ for $n \ge 2q-2$.
\end{proof}

\noindent Uniqueness follows because $|P| < 1$ off the atoms, so any minimizer lives on the $q$ atoms,
where its first $q$ moments determine it (Appendix~\ref{app:proofs_resonance}).

Structurally the optimum is again a discrete Fourier transform, of the periodized data: a transform on
the $2q$-point grid supported on $q$ of its points (the odd ones when $\sigma = 1$), rather than one of
length $N$. That is the whole content of a resonance: once $u_{k+q} = \pm u_k$ the $N$ moment constraints
are redundant, and the arithmetic says exactly how redundant. The construction of \cite{anchor}, built
for arbitrary data, inverts a square $N \times N$ system and need not use that periodicity; the optimal
sampler is the same transform, taken short. This Fourier structure of the \emph{answer} is a conclusion,
not a hypothesis: Definition~\ref{def:cost} minimizes over all complex measures on the circle, with no
grid, and the minimizer collapses onto $q$ equally spaced points only at a resonance.

\paragraph{Below the threshold.} Theorem~\ref{main:B}(ii) leaves a finite-$n$ question, the threshold $n_0$. Since adding moment constraints cannot lower the cost, $\Gmin < q$ for $n < n_0$ and the $q$-point sampler is strictly beatable there. The certificate puts $n_0 \le 2q-2$; a Parseval argument on the degree-$(q-1)$ interpolant puts $n_0 \ge q$ \cite[Rem.~4.9]{v1}; and certified two-sided computation puts the $p = 1$ chirp at $n_0 \in \{\ceil{3q/2}-1, \ceil{3q/2}\}$ for every $q \le 20$ and every prime $q \le 41$ \cite[Rem.~4.10 and Table~1]{v1}. Both ends of the bracket $q \le n_0 \le 2q-2$ are attained within the class of sequences for which the argument above works verbatim, the bi-unimodular (CAZAC) sequences, so the threshold is a property of the individual sequence, not of the class \cite[Thm.~B(iii) and Rems.~4.8--4.10, in that paper's numbering]{v1}; what decides it is open.

\section{\texorpdfstring{Irrational angles: the continued-fraction dichotomy\\ (\mbox{Theorems} C and~D)}{Irrational angles: the continued-fraction dichotomy (Theorems C and D)}}\label{sec:generic}

The resonances are exact at the rungs. Off them, the rung's certificate survives a small detuning and the rung's measure can be convolved with a cheap correction; together, these put the cost near a rung of denominator $q$ at $\Theta(q)$ as long as the detuning is below $1/n^2$. At every register size there is a top convergent of $\gamma/\pi$ with denominator at most $(n+2)/2$, and whether its detuning is small is decided by the next partial quotient. That is the whole mechanism of the dichotomy; this section carries it out, and ends with the master formula (Theorem~\ref{main:D}), which reads the cost off the nearest rung.

\begin{lemma}[submultiplicativity of the LCU norm]\label{lem:submult}
For unimodular data sequences $u, v$ on $k = 0, \ldots, n$:
$\sqrt{\Gmin(uv)} \le \sqrt{\Gmin(u)}\cdot\sqrt{\Gmin(v)}$.
\end{lemma}

\begin{proof}
If $\hat\mu = u$ and $\hat\nu = v$ on $k \le n$, the convolution $\mu * \nu$ has
$\widehat{\mu * \nu} = \hat\mu\,\hat\nu = uv$ there, and
$\|\mu * \nu\|_{\TV} \le \|\mu\|_{\TV}\|\nu\|_{\TV}$. Minimize both factors.
\end{proof}

Chirp data multiply in the angle, $e^{-i\gamma_1 k^2} e^{-i\gamma_2 k^2} = e^{-i(\gamma_1 + \gamma_2)k^2}$,
so the chirp cost is submultiplicative in $\gamma$-addition; in particular, for a detuned resonance,
\begin{equation}\label{eq:submult}
\Gmin(\pi p/q + \delta,\, n) \;\le\; q \cdot \Gmin(\delta, n).
\end{equation}

\begin{corollary}[detuned upper bound, explicit]\label{cor:detunedUB}
For every resonance $\pi p/q$ and every detuning $\delta$,
\[
\Gmin(\pi p/q + \delta,\, n) \;\le\; q \cdot \exp\Big( |\delta|\,\frac{N^2 - 1}{3} \Big).
\]
\end{corollary}

\noindent The proof (Appendix~\ref{app:proofs_master}) bounds $\sqrt{\Gmin(\delta, n)}$ by the
Wiener-algebra norm of $e^{-i\delta f}$ for the \emph{centred} parabola $f(k) = (k - n/2)^2$; centring
is what keeps the bound informative on the whole window $|\delta| \lesssim 1/n^2$ with no logarithmic
loss. The matching lower bound rotates the rung's certificate. Write
\[
\Psi_q(\delta) \;:=\; \frac1q\Big|\sum_{|m| < q}\Big(1 - \frac{|m|}{q}\Big)e^{-i\delta m^2}\Big|,
\qquad \Psi_q(0) = 1, \qquad \Psi_q(\delta) \;\ge\; 1 - \frac{7\,\delta^2q^4}{360}
\]
(Lemma~\ref{lem:fejerdetune}, Appendix~\ref{app:proofs_master}).

\begin{theorem}[detuned lower bound: the resonance floor]\label{thm:detunedLB}
Let $\gamma = \pi p/q + \delta$ with $\gcd(p,q) = 1$, and let $n \ge 2q-2$. Then
\[
\sqrt{\Gmin(\gamma, n)} \;\ge\; \sqrt q\,\Psi_q(\delta) \;\ge\; \sqrt q\,\Big(1 - \frac{\delta^2q^4}{30}\Big).
\]
\end{theorem}

\noindent The rung's Fej\'er certificate, rotated and rephased by Lemma~\ref{lem:free}, stays admissible
and pairs to $\sqrt q\,\Psi_q(\delta)$; its support is the only use of $n \ge 2q-2$
(Appendix~\ref{app:proofs_master}). The right-hand side does not depend on $n$: the bound is a permanent
floor, valid at every $n$ above the threshold, with natural window $|\delta| \lesssim 1/q^2$.

At a given $n$ the floor is used at the top convergent, where the detuning is controlled by the next
partial quotient.

\begin{theorem}[the convergent floor]\label{thm:convfloor}
Let $x = \gamma/\pi$ have convergents $p_j/q_j$ and partial quotients $a_j$, let $n \ge 2$, and let $k$ be
the largest index with $q_k \le (n+2)/2$. Put $\delta_k = \gamma - \pi p_k/q_k$, $t_k = |x - p_k/q_k|\,q_k^2$ and $a^* = a_{k+1}$. Then
\[
\Gmin(\gamma, n) \;\ge\; q_k\,\Psi_{q_k}(\delta_k)^2 \;\ge\; q_k\big(1 - \beta t_k^2\big)^2 \;\ge\; q_k\Big(1 - \frac{\beta}{a^{*2}}\Big)^2 \;\ge\; 0.6530\,q_k ,
\qquad \beta := \frac{7\pi^2}{360},
\]
and $q_k > (n+2)/(2(a^*+1))$; hence $\Gmin(\gamma,n) \ge c(a^*)\,(n+2)$ with
$c(a) = (1 - \beta/a^2)^2/(2(a+1))$.
\end{theorem}

\noindent Theorem~\ref{thm:detunedLB} applies at the top convergent automatically, since
$q_k \le (n+2)/2$ is $n \ge 2q_k - 2$; its approximation quality bounds the detuning factor, because
$|x - p_k/q_k| < 1/(q_kq_{k+1})$ and $q_{k+1} \ge a^*q_k$; and maximality of $k$ gives
$q_{k+1} > (n+2)/2$ with $q_{k+1} \le (a^*+1)q_k$ (Appendix~\ref{app:proofs_generic}). The floor is a
statement about a single Diophantine quantity, the partial quotient at scale $n$, and the dichotomy
follows.

\begin{mainthm}[dichotomy]\label{main:C}
Let $x = \gamma/\pi$ have partial quotients $a_j$.
\begin{enumerate}[label=(\roman*), nosep]
\item $\Gmin(\gamma, n) = \Theta(n)$ as $n \to \infty$ if and only if $x$ is badly approximable. If the
$a_j$ are unbounded then $\liminf_n \Gmin(\gamma,n)/n = 0$: at the scales
$n_k = \lfloor\sqrt{q_kq_{k+1}}\rfloor - 1$, $\Gmin(\gamma, n_k)/n_k \le 2.85\,(1 + 2/n_k)/\sqrt{a_{k+1}}$.
\item If $a_j \le A$ for all $j$ then, for all $n \ge 2$, $c_A\,(n+2) \le \Gmin(\gamma, n) \le n+1$ with
the explicit constant $c_A = c(A)$ of Theorem~\ref{thm:convfloor} ($c_1 = 0.163$).
\item $\sup_n \Gmin(\gamma, n) < \infty$ if and only if $x$ is rational; at $\gamma = \pi p/q$ in lowest
terms the supremum is exactly $q$, attained at every $n \ge 2q-2$, and at irrational $x$,
$\Gmin(\gamma, n) \to \infty$.
\item For one bosonic mode with no photon-number cutoff --- the collective observable the photon number
$\hat n$, the resource family the linear phase shifters $e^{i\theta\hat n}$, the target the Kerr phase
$e^{-i\gamma\hat n^2}$, the moment conditions running over every $k \ge 0$ --- the least overhead $\Gamma_\infty(\gamma)$ of a representation $e^{-i\gamma\hat n^2} = \int e^{i\theta\hat n}\,d\mu(\theta)$ by a complex measure is $\sup_n\Gmin(\gamma, n)$. Hence a finite-cost representation exists if and only if $x$ is rational; at
$\gamma = \pi p/q$ its least cost is exactly $q$, attained only by the resonant measure $\mu^*$ of
Theorem~\ref{main:B}; and at irrational $x$ the cost at photon cutoff $n$ grows without bound, at the rates of (i).
\item For almost every $\gamma$ and every $\varepsilon > 0$, $\Gmin(\gamma, n) \ge c\,n/(\log n\,(\log\log n)^{1+\varepsilon})$ for all large $n$, while $\Gmin(\gamma, n_k) \le 3.11\,(1+o(1))\,n_k/\sqrt{\log n_k \log\log n_k}$ for infinitely many $k$.
\end{enumerate}
\end{mainthm}

\begin{proof}[Proof of (i) and (ii)]
``$\Leftarrow$'' and (ii): Theorem~\ref{thm:convfloor} with $a^* \le A$; on the integers $c(a)$ is
decreasing, so the worst case is $a^* = A$, and $\Gmin \le n+1$ is Lemma~\ref{lem:dft}.
``$\Rightarrow$'': if the $a_j$ are unbounded, take $k$ with $a_{k+1}$ large and the scale $n_k = \lfloor\sqrt{q_kq_{k+1}}\rfloor - 1$. The convergent's detuning is $|\delta_k| = \pi|x - p_k/q_k| < \pi/(q_kq_{k+1})$, so Corollary~\ref{cor:detunedUB} gives $\Gmin(\gamma, n_k) \le q_k\exp\bigl(\pi(n_k+1)^2/(3q_kq_{k+1})\bigr) \le e^{\pi/3}q_k$, while $q_k \le \sqrt{q_kq_{k+1}/a_{k+1}} \le (n_k+2)/\sqrt{a_{k+1}}$; hence $\Gmin(\gamma, n_k)/n_k \le e^{\pi/3}(1 + 2/n_k)/\sqrt{a_{k+1}}$, which tends to zero along the subsequence (Appendix~\ref{app:proofs_generic}).
\end{proof}

\noindent Part (iii) is Theorem~\ref{main:B} at rational angles and, at irrational ones, the floor $0.6530\,q_k$ of Theorem~\ref{thm:convfloor} growing with the convergents. Part (iv): the truncated costs
increase with $n$, a weak-$*$ limit of near-optimal measures matches every moment, and (iii) with the
uniqueness of Theorem~\ref{main:B} gives the rest (Appendix~\ref{app:proofs_generic}). In a register
every target has a finite decomposition and rationality decides only whether its cost stays bounded;
in a mode it decides whether a decomposition exists at all, and the exactness threshold disappears. Part (v) is Borel--Bernstein and L\'evy's law applied to Theorem~\ref{thm:convfloor} and part (i); it uses nothing else (Appendix~\ref{app:proofs_generic}).

\paragraph{The master formula.} Rational angles cost their denominator; badly approximable angles cost
$\Theta(n)$; in between, the cost is read off the nearest low rung. Write
\[
F(\gamma, n) \;:=\; \min_{\gcd(p,q) = 1} q\,\bigl(1 + |\gamma - \pi p/q|\,n^2\bigr),
\qquad\text{capped at } n+1 ,
\]
the cost of the nearest rung times one plus the detuning in units of $1/n^2$.

\begin{conjecture}[master formula]\label{conj:masterv2}
There are absolute constants $C_1, C_2 > 0$ with $C_1 F(\gamma, n) \le \Gmin(\gamma, n) \le C_2 F(\gamma, n)$ for every $\gamma$ and every $n$ \cite[Conj.~5.1]{v1}.
\end{conjecture}

\noindent It is a theorem on the \emph{deep-resonance range}, within $1/n^2$ of a rung of denominator at most $(n+2)/2$:

\begin{mainthm}[master formula on the deep-resonance range]\label{main:D}
Let $n \ge 8$, and let $q^*$ be the smallest denominator carrying a rung inside the window
$|\gamma - \pi p/q| \le 1/n^2$. For every $\gamma$ with $q^* \le (n+2)/2$,
\[
0.990\, q^* \;\le\; \Gmin(\gamma, n) \;\le\; 1.518\, q^*,
\]
and, since $q^* \le F(\gamma, n) \le 2q^*$ on this range, $0.49\,F(\gamma, n) \le \Gmin(\gamma, n) \le 1.52\,F(\gamma, n)$.
\end{mainthm}

\noindent The hypothesis $q^* \le (n+2)/2$ is $n \ge 2q^*-2$, so Theorem~\ref{thm:detunedLB} applies at
the window rung and Corollary~\ref{cor:detunedUB} bounds above (Appendix~\ref{app:proofs_master}). Off
the deep-resonance range a proved reduction leaves the conjecture at one two-sided growth statement at
the top convergent, $\Gmin(\pi p_k/q_k + \delta, n) \asymp \min(q_k(1 + |\delta|n^2), n)$, of which neither half is known on the shoulder, the range $|\delta|n^2 \gg 1$ where the detuning is large on the scale $1/n^2$ but the rung is still the nearest \cite[Lemma~5.8, Open Problem~M.1]{v1}.

\paragraph{Two remarks on ``generic''.} First, at a fixed angle: the linear law holds only on the badly approximable set, which has measure zero, and by part (v) ``$\Gmin \gtrsim n$ beyond the badly approximable angles'' is false for almost every angle; the cost is $n$ up to logarithmic factors, with dips along infinitely many scales. Second, at a fixed $n$: most angles cost a constant fraction of $n+1$.

\begin{proposition}[typical angles at fixed $n$]\label{prop:typical}
Let $x = \gamma/\pi$ be uniform on $[0,1)$. For every $n \ge 2$ and every $c > 0$,
\[
\bigl|\{x : \Gmin(\pi x, n) < c\,(n+1)\}\bigr| \;\le\; 6.13\,c ,
\]
and for every $f \ge 1$ with $(n+1)/f \ge 2$,
\[
\bigl|\{x : \Gmin(\pi x, n) \le (n+1)/f\}\bigr| \;\ge\; 0.290/f^2 .
\]
In particular the median of $\Gmin(\pi x, n)/(n+1)$ is at least $0.122$ for every $n \ge 2$, and at least
$0.19$ for every $n \ge 6$. The constant $0.290$, and the bound $0.19$ for $n \le 6000$, are computer-assisted.
\end{proposition}

\noindent The proof (Appendix~\ref{app:proofs_generic}) bounds the cost below by its top convergent
alone and bounds the measure of the angles whose top convergent has a given denominator; the share of
cheap angles is Corollary~\ref{cor:detunedUB} on disjoint windows around the rungs with
$q \le (n+1)/f$. The proved median $0.12$ is a factor five below the computed $0.6$ of
Figure~\ref{fig:landscape}, and the certificates of this paper cannot close the gap: a single-rung
certificate pairs to at most $\sqrt q$ with $2q-2 \le n$, so it proves at most $(n+2)/2$ at any angle.
A sharper constant at badly approximable angles --- $\Gmin \ge N/C(\gamma)^2$ with $C$ measured in $[1.53, 2.07]$ at
the golden and silver angles --- follows from theta-sum flatness through the matched
filter~\eqref{eq:matched}, by Hardy--Littlewood's uniform $O(\sqrt n)$ bound at bounded partial quotients
\cite{hl1914} as recovered by Cosentino--Flaminio \cite{cf15} \cite[Thm.~6.8]{v1}; the explicit constants
above use no theta-sum input.

\section{Nothing symmetric helps: independent angles, single-qubit gates, diagonal Cliffords (Theorems E and F)}\label{sec:enlarge}

Two enlargements of the resource set test whether the arithmetic is an artifact of a small
dictionary. Independent per-qubit angles, and then arbitrary single-qubit gates, change nothing;
adjoining the symmetric diagonal Clifford phases changes the ladder by an exact $2$-adic rule rather
than erasing it. Both keep the permutation symmetry; \S\ref{sec:separation} removes it.

\subsection{Independent angles and single-qubit gates}\label{sec:torus}

Definition~\ref{def:cost} fixes one common angle. Single-qubit layers naturally carry independent
angles $\vec\varphi \in \mathbb{T}^n$, acting as $V(\vec\varphi)|x\rangle = e^{i\varphi \cdot x}|x\rangle$.

\begin{mainthm}[independent angles and single-qubit gates share the same cost]\label{main:E}
\begin{enumerate}[label=(\roman*), nosep]
\item Let $U$ be diagonal with data $u_x = u_{|x|}$ depending only on Hamming weight (any profile), and let $\Gamma_{\mathbb{T}}(U)$ be the minimal overhead over complex measures on the full torus $\{V(\vec\varphi)\}$. Then $\Gamma_{\mathbb{T}}(U) = \Gmin(U)$.
\item Let $\Gamma_{\otimes}(U)$ be the least overhead of an LCU of a diagonal unitary $U$ into arbitrary product unitaries $W = w_1 \otimes \cdots \otimes w_n$, i.e.\ one layer of single-qubit gates. Then $\Gamma_{\otimes}(U) = \Gamma_{\mathbb{T}}(U)$, with or without symmetry.
\end{enumerate}
\end{mainthm}

\noindent \textit{Proof idea for (i).} Symmetrize the certificate: for symmetric multiaffine polynomials the Grace--Walsh--Szeg\H{o} theorem \cite{rahmanschmeisser} equates the sup-norm on $\mathbb{T}^n$ with that of a one-variable polynomial on the circle (Appendix~\ref{app:proofs_enlarge}). So every lower bound in this paper for a symmetric target holds verbatim over arbitrary independent per-qubit angles, and by (ii) over arbitrary single-qubit gates.

\begin{proof}[Proof of (ii)]
$\le$ is inclusion of resource sets. For $\ge$, apply the pinching
$\Delta(X) = \sum_x |x\rangle\langle x|X|x\rangle\langle x|$ to $U = \sum_j c_jW_j$: $\Delta(U) = U$, and
$\Delta(W_j) = \bigotimes_i \Delta(w_{j,i})$ with each $\Delta(w) = \mathrm{diag}(w_{00}, w_{11})$ a
diagonal contraction. Each entry lies in the closed unit disc and is therefore the midpoint of two
unimodular numbers, so $\Delta(w)$ is the average of four diagonal unitaries, and $\Delta(W_j)$ is a
convex combination of diagonal product unitaries $e^{i\varphi_0}V(\vec\varphi)$. Substituting does not
increase $\sum_j|c_j|$, and the global phases are absorbed into the coefficients.
\end{proof}

\noindent The pinching is the completely dephasing channel, a uniform mixture of diagonal unitaries
\cite[\S2.6.3]{tomamichel2016}; the argument is the usual one that a free map fixing the target and
sending the dictionary into its own convex hull cannot lower an atomic norm. We have not found it stated
for gate decompositions. For a bosonic mode the reduction does not go through --- a pinched squeezer has
vacuum entry below one --- so whether Gaussian branches can beat phase shifters for the Kerr gate is
open, and \S\ref{sec:pricing_kerr} is stated over phase shifters.

\subsection{Diagonal Cliffords: the ladder transforms \texorpdfstring{$2$}{2}-adically}\label{sec:sheets}

The natural \emph{entangling} enlargement changes the ladder by an exact arithmetic rule: the resonance
structure is compatible with the standard gate hierarchy, not an artifact of an impoverished gate set.
Adjoin to $S$ the permutation-symmetric diagonal Clifford-hierarchy phases --- collective $S$, $T$,
\dots\ gates and the all-pairs $CZ$, $CS$, \dots\ phases. At level $\ell+1$ of the hierarchy these
contribute $2^\ell$ quadratic \emph{sheets}: unit-cost atoms $e^{ick^2}$ with $c \in \pi\mathbb{Z}/2^\ell$.
Indexing by $\ell$, write $D(\ell)$ for the enlarged dictionary and $\Gamma^{D(\ell)}$ for its cost.

\begin{mainthm}[$2$-adic sheet map]\label{main:F}
Let $\gamma = \pi p/q$ in lowest terms and let $v_2(q)$ be the $2$-adic valuation of $q$. Write
$Q_{\min} = q/2^{v_2(q)}$ if $v_2(q) \le \ell$ and $Q_{\min} = q$ otherwise. Then
\[
\Gamma^{D(\ell)}(\pi p/q,\, n) \;=\; Q_{\min}
\qquad\text{for all }\;
n \;\ge\; \begin{cases}
2^{\ceil{\ell/2}}\,(2Q_{\min}-2), & v_2(q) \le \ell,\\[2pt]
2q-2, & v_2(q) > \ell,
\end{cases}
\]
and $\Gamma^{D(\ell)} \le Q_{\min}$ for every $n$.
\end{mainthm}

\noindent In words: no sheet depth below $v_2(q)$ helps at all, depth $v_2(q)$ removes the entire
$2$-power at once, and odd denominators are immune at every fixed depth. The best sheet lands on a rung
of denominator $Q_{\min}$; conversely the Fej\'er certificate, thinned onto $2^{\ceil{\ell/2}}\mathbb{Z}$
when $v_2(q) \le \ell$, is admissible on every sheet at once --- thinned, not lengthened, which is why
the degree is $2^{\ceil{\ell/2}}(2Q_{\min}-2)$ rather than $2^\ell Q_{\min}$
(Appendix~\ref{app:proofs_enlarge}). Uniqueness does not transfer: for $v_2(q) > \ell$ each of the
$2^\ell$ sheets realises the optimum on its own. Off the rungs the plateau survives at every fixed
depth, and the ladder collapses only when the dictionary contains a near-copy of the target, at depth
$\ell \approx 2\log_2 n$, since the nearest sheet has residual rate at most $\pi/2^{\ell+1}$ and
Corollary~\ref{cor:detunedUB} at the $q = 1$ rung gives
$\Gamma^{D(\ell)} \le \exp(\pi(N^2-1)/(3 \cdot 2^{\ell+1})) \to 1$ once $2^\ell \gg n^2$
\cite[\S8.2, App.~I]{v1}.

\section{\texorpdfstring{Without symmetry: an exponential cost no single cut\\ \mbox{accounts} for (Theorem G)}{Without symmetry: an exponential cost no single cut accounts for (Theorem G)}}\label{sec:separation}

The Fourier-based construction is efficient because of symmetry: its guarantee $\Gamma \le N$ is
polynomial in the number $N$ of collective levels, and permutation symmetry holds that number at
$n+1$. A binary register has $N = 2^m$, and Theorem~\ref{main:E}, which admits independent per-qubit
angles at no change in cost, needs the symmetry too. This section reaches the full torus for a natural
target without symmetry --- the squaring phase of Fourier arithmetic \cite{draper2000, rpge2017} on a
binary register --- and the answer is a separation: two-local diagonal gates implement it at overhead
one, and no randomized compilation over single-qubit gates does so below an exponential in the register
size, neither as an LCU nor, by Remark~\ref{rem:productform}, as an ancilla-free sampler.

\paragraph{The cut bound.} When the target is not permutation-symmetric, Theorem~\ref{main:E} is silent,
but an elementary tool still reaches the torus. For a bipartition $A|B$ of the $m$ qubits, reshape the
diagonal data $g(x)$ into the $2^{|A|}\times 2^{|B|}$ matrix $G_{A|B}[x_A, x_B] := g(x_A, x_B)$.

\begin{lemma}[cut bound]\label{lem:cut}
For every diagonal $g$ on $m$ qubits and every bipartition $A|B$,
\[
\sqrt{\Gamma_{\mathbb{T}}(g)} \;\ge\; \|G_{A|B}\|_{\mathrm{tr}}\big/2^{m/2}.
\]
\end{lemma}

\begin{proof}
A resource element $e^{i\varphi\cdot x} = e^{i\varphi_A\cdot x_A}e^{i\varphi_B\cdot x_B}$ reshapes to
the rank-one matrix $v_Av_B^{\top}$ with unimodular entries, whose trace norm is
$\|v_A\|_2\|v_B\|_2 = 2^{m/2}$ for every cut. If $g = \sum_a c_a e^{i\varphi_a\cdot x}$, the triangle
inequality gives $\|G_{A|B}\|_{\mathrm{tr}} \le 2^{m/2}\sum_a|c_a|$. Equivalently, with $G = W|G|$ the
polar decomposition, $a = \overline{W}/2^{m/2}$ is an admissible torus certificate ---
$|A(z)| = |v_A(z)^{\top}\overline{W}v_B(z)|/2^{m/2} \le 1$ on $\mathbb{T}^m$ --- pairing to
$\mathrm{tr}(W^{\dagger}G)/2^{m/2} = \|G\|_{\mathrm{tr}}/2^{m/2}$.
\end{proof}

\noindent Because the projectors $|x_A\rangle\langle x_A|$ are Hilbert--Schmidt orthonormal,
$\|G_{A|B}\|_{\mathrm{tr}}$ is the operator-Schmidt trace norm of $\mathrm{diag}(g)$ across the cut
\cite{nielsen2003}: the sampling cost over product diagonal unitaries is bounded below by the operator
entanglement of the target, normalised by that of a product unitary, and the bound may be taken over
all cuts. No cut certifies more than $2^{\lfloor m/2\rfloor}$ for any diagonal target: $G_{A|B}$ has
unimodular entries and rank at most $2^{\min(|A|,|B|)}$, so its trace norm is at most
$2^{(m+\min(|A|,|B|))/2}$. For two-qubit gates the cut bound is the whole story, since the one-layer
product cost equals the Choi--Schmidt bound whenever the Schmidt operators are unitary
\cite[Prop.~3.8]{hl2025}, \cite[Cor.~3.1]{sps2025}; for larger gates it is, to our knowledge, the only
lower bound in use \cite{necaise2026}.

\begin{mainthm}[register separation]\label{main:register}
Let $f_m(x) = e^{2\pi i k(x)^2/2^m}$ with $k(x) = \sum_{j<m}2^jx_j$ be the squaring phase of Fourier
arithmetic on an $m$-qubit register, and $\Gamma_{\mathbb{T}}$ its cost over the full torus of
independent per-qubit diagonal angles, equal by Theorem~\ref{main:E}(ii) to its cost over arbitrary single-qubit gates.
\begin{enumerate}[label=(\alph*), nosep]
\item For every $m \ge 2$,
\[
2^{\lfloor (m-1)/2\rfloor} \;\le\; \Gamma_{\mathbb{T}}(f_m) \;\le\; 2^{m-2} \;<\; 2^{m-1} \;=\; \Gmin(f_m),
\]
and the lower bound is the best any single cut gives for this target.
\item For every $m \ge 2$,
\[
\Gamma_{\mathbb{T}}(f_m) \;\ge\; 2^{\lfloor (m-1)/2\rfloor}\,\bigl(4-2\sqrt2\bigr)^{\lfloor (m-2)/4\rfloor},
\]
strictly above the cut bound for every $m \ge 6$. Hence
$\liminf_m \frac1m\log_2\Gamma_{\mathbb{T}}(f_m) \ge \frac14\log_2\bigl(8(2-\sqrt2)\bigr) = 0.557\ldots > \frac12$.
\end{enumerate}
\end{mainthm}

\noindent Part (a): below, across the low-bits/high-bits cut the phase is a flat-spectrum Fourier matrix up to diagonal unitaries, and no cut certifies more because the top bit is inert; above, peeling off bit~$0$ costs a factor $2$ in $\sqrt{\Gamma_{\mathbb{T}}}$ and leaves $f_{m-2}$; the one-parameter value $\Gmin(f_m) = 2^{m-1}$ is Theorem~\ref{main:B} at the rung $q = 2^{m-1}$, since $n = 2^m - 1 \ge 2q-2$ (Appendix~\ref{app:proofs_separation}). The small cases are exact: $\Gamma_{\mathbb{T}}(f_m) = 1,\ 2,\ 2+\sqrt2$ at $m = 2, 3, 4$ \cite[Rem.~H.2]{v1}. Part (b) \emph{lifts}
a certificate: from a certificate for $f_{m-4}$ it builds one for $f_m$, by gluing four twisted copies
of the smaller certificate, one per \emph{slot}, through two \emph{glue bits} (the outer bits $m-3$ and
$m-2$) with \emph{readout polynomials} $Q_k$ in those bits' torus variables. In one line: one glue bit reproduces the cut-bound rate exactly; two glue bits read out jointly open a strict gap between the $\ell^2$ identity the readout polynomials satisfy and the $\ell^1$ price a certificate pays; and iterating that lift gives the exponent $0.557\ldots$

\begin{proof}[Proof sketch of (b); the full proof is in Appendix~\ref{app:proofs_separation}]
The top bit is inert. Splitting $k = k_L + 4c + 2^{m-3}y + 2^{m-1}x_{m-1}$, with $k_L$ and $y$ carried by two low and two outer bits and $c$ by the block in between, writes the target, for each fixed $k_L$, as a copy of $f_{m-4}$ on the middle block, twisted by a one-local rotation and multiplied by the character $\zeta_{k_L}$ of $y$, where $\zeta_k = (i^k, (-1)^k) \in \mathbb{T}^2$. An optimal certificate for $f_{m-4}$ therefore lifts to one for $f_m$ that glues the four twisted copies with readout polynomials $Q_k$ in the outer variables, here the two-bit Fourier projectors $Q_k(w) = \frac14\prod_{t=0,1}\bigl(1 + e^{-2\pi ik2^t/4}w_t\bigr)$, which satisfy $Q_k(\zeta_k) = 1$. The pairing grows by $\sum_kQ_k(\zeta_k) = 4$ and the sup-norm by at most $\sup_{\mathbb{T}^2}\sum_k|Q_k| = 2\cos\frac\pi8$, so $\sqrt{\Gamma_{\mathbb{T}}(f_m)} \ge 2\sec\frac\pi8\,\sqrt{\Gamma_{\mathbb{T}}(f_{m-4})}$, and $4\sec^2\frac\pi8 = 4(4-2\sqrt2)$. Iterating down to $m \le 5$ and finishing with the cut bound gives the statement.
\end{proof}

No cut bound can do this (Figure~\ref{fig:separation}): Lemma~\ref{lem:cut} never certifies more than $2^{\lfloor m/2\rfloor}$, and the
squaring phase, whose top bit is inert, reaches that ceiling on its other $m-1$ qubits. Part (b) goes
past it: the sampling cost of $f_m$ is not accounted for by operator entanglement across any single
bipartition. What the proof uses instead is the scale structure of the phase, four twisted copies of
$f_{m-4}$ glued by a two-bit Fourier transform. The splitting is the defining property of a quadratic
phase function \cite[eq.~(7)]{bvvdn2014}, the reduction step of the prime-power recursion for quadratic
Gauss sums kept as an identity instead of summed; the gluing has the form of the Golay--Rudin--Shapiro
recursion \cite{shapiro1951, rudin1959} and its paraunitary generalization \cite{ppt2004}, and the lift
is in spirit the product-witness argument behind the multiplicativity of extents on tensor products
\cite[Prop.~1]{bravyi2019}, applied to a target that is a \emph{sum} of twisted copies rather than a
product.

The gain comes from where the lift parts ways with the complementary-sequence recursion it resembles. The two-bit Fourier projectors satisfy the paraunitary
identity $\sum_k|Q_k(w)|^2 = 1$ on $\mathbb{T}^2$, so by Cauchy--Schwarz $\sum_k|Q_k| \le 2$, which would
return exactly the cut rate; equality needs all four $|Q_k(w)|$ to be equal at one point, and no point of
$\mathbb{T}^2$ has that, so $\sup\sum_k|Q_k| = 2\cos\frac\pi8 < 2$. Concretely, with
$r(w) := \frac12(|1+w| + |1-w|)$ the one-glue-bit sum, $|Q_0| + |Q_2| = \frac12|1+w_1|\,r(w_0)$ and
$|Q_1| + |Q_3| = \frac12|1-w_1|\,r(-iw_0)$, so $\sup_{w_1}\sum_k|Q_k| = (r(w_0)^2 + r(-iw_0)^2)^{1/2}$: it
could reach $2$ only where both copies of $r$ reach $\sqrt2$, and they do so a quarter turn apart, at $w_0 = \pm i$ and at $w_0 = \pm1$. Complementary sets are built for the $\ell^2$
identity; a dual certificate pays in $\ell^1$, and the lift gains exactly from the strict gap between the
two norms. Same gluing, different norm, and the difference is the whole separation. A lift through one
glue bit grows at exactly the cut bound's rate; it takes two glue bits read out jointly to get past it.

\begin{remark}[how far lifting can go]\label{rem:liftcap}
Splitting off $j$ pairs of outer bits and optimizing the readout polynomials raises the gain per $2j$
bits, but for every choice
\[
\Bigl|\sum_kQ_k(\zeta_k)\Bigr| \;\le\; \sqrt2\,(\pi/2)^{j-1}\,\sup\sum_k|Q_k|
\]
(Appendix~\ref{app:proofs_separation}), so no lift of this kind certifies more than
$\Gamma_{\mathbb{T}}(f_m) \gtrsim 2^{m\log_2(\pi/2)} = 2^{0.65\,m}$. The measured exponent below is
beyond it: a lift builds every slot from rotations of one certificate, and going further needs a joint
constraint on different certificates \cite[Prop.~H.1, Open Problem~M.3]{v1}.
\end{remark}

\begin{figure}[t]\centering
\includegraphics[width=0.52\textwidth]{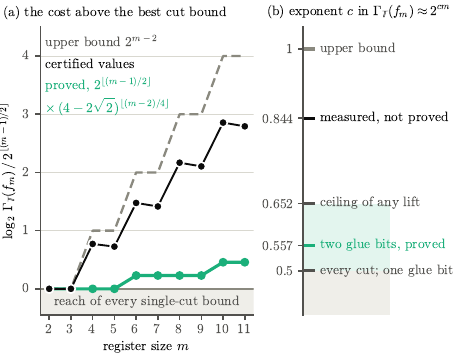}
\caption{The register cost against every cut bound. (a)~$\Gamma_{\mathbb{T}}(f_m)$ divided by $2^{\lfloor (m-1)/2\rfloor}$, the most any single-cut bound certifies for this target: certified values
through $m = 11$ and the bounds of Theorem~\ref{main:register}. (b)~The growth exponents; the shaded
band is the reach of lifting (Remark~\ref{rem:liftcap}).}
\label{fig:separation}
\end{figure}

\begin{conjecture}[the exponent]\label{conj:exponent}
The limit $\kappa := \lim_m \Gamma_{\mathbb{T}}(f_m)^{1/m}$ exists and satisfies $\kappa < 2$: the
$\Theta(2^m)$ plateau of the one-parameter theory does not survive on the torus.
\end{conjecture}

\noindent By Theorem~\ref{main:register}(b) any limit is at least $2^{0.557}$, strictly above the cut bound's $\sqrt2$; the
conjecture is about the rest. Two-sided brackets of relative width $10^{-5}$ through $m = 9$ and
certified values through $m = 11$, obtained from an identity that writes the level-$m$ dual as a
two-certificate problem at level $m-2$ \cite[Prop.~H.1]{v1}, give two-step ratios
$\sqrt{\Gamma_{m}/\Gamma_{m-2}}$ that converge geometrically within each parity class, to $1.7955$ and
$1.7953$: extrapolated, $\kappa = 2^{0.844 \pm 0.001}$ (Figure~\ref{fig:separation}). These are
measurements, not a proof.

The separation is between implementation resource sets, not complexity classes: the squaring phase is a
normalizer quadratic phase gate and classically simulable \cite{bvvdn2014}, and what is priced is the
exchange rate of \S\ref{sec:intro}. Nor is it operator entanglement in disguise, since it exceeds every
cut. For the two-qubit phases behind its small cases the cut bound is tight,
$\Gamma_{\mathbb{T}}(CP(\varphi)) = 1 + \sin(\varphi/2)$, which is Piveteau's regularized
quasi-probability extent under local operations \cite{piveteau2025thesis} recovered in a different
model; the single-shot value with classical communication is $2\Gamma_{\mathbb{T}} - 1$ \cite{ps2024},
$3$ for $CZ$ against our $2$ \cite[Rem.~H.2]{v1}.

\section{What the theory prices}\label{sec:pricing}
\label{sec:anchorconsequences}

Read left to right, the law takes an angle and returns its cost: in the optimization reading the angle
is a design variable, and the objects to price are the constructions practitioners run. Read right to
left, the physics fixes the angle, and the cost describes the state the evolution produces.

\subsection{The Fourier-based LCU}\label{sec:pricing_lcu}

The equispaced $N$-atom construction of \cite{anchor} --- one rotation angle per collective level,
weights from an inverse DFT of the data, $\Gamma \le n+1$ (Lemma~\ref{lem:dft}) --- is run by its
authors inside penalty-based QAOA for the densest-$k$-subgraph problem, where each QAOA layer contains
one penalty phase $e^{-i\gamma K^2}$: exact statevector simulations on a $12$-node random $3$-regular
graph with $k = 4$, and a $106$-qubit hardware run \cite[Sec.~VI]{anchor}.

\paragraph{Where it is optimal.} Its guarantee $\Gamma \le n+1$ is tight up to a constant factor at a
typical angle: at the median over angles the optimum is about $0.6(n+1)$, in certified computation at
$n = 8, 16, 32, 64$ and $128$ (the median lies in $[0.60, 0.62]$ times $n+1$ at every one of them), and
at least $0.12(n+1)$ by Proposition~\ref{prop:typical}. Both are statements about their \emph{sampler},
not only about decompositions: by Proposition~\ref{prop:domination}, $\Gmin$ is the least constant with
which any sampling distribution over the rotations meets their guarantee $\tilde p_x \ge p_x/\Gamma$ for
every downstream circuit. The remaining slack is in the construction's actual cost rather than in its
bound: at $n = 32$ the equispaced measure costs a median $1.36$ times the optimum, and on the $12$-node
instance of \cite{anchor} the optimum saves about $9\%$ at QAOA depth one, $\Gmin \in [10.3, 10.5]$
against the DFT's $11.5$, and more than a third at depth two, $[89.9, 93.0]$ against $146.6$, where the
overheads of the two penalty layers multiply (certified brackets; \cite[\S7.1, Rem.~7.1]{v1}). The
construction has one virtue the optimality theory does not supply: it is a single closed-form formula
that works for every diagonal collective target, with no arithmetic input and no per-angle
optimization, and at a typical angle no decomposition improves on it by more than that constant.

\paragraph{At the resonances: an $n$-independent recipe.} At $\gamma = \pi p/q$ the theory identifies
what the construction does not use: a closed-form sampler over $q$ equally spaced angles with explicit
phases (Theorem~\ref{main:B}), at overhead exactly $q$ for every $n \ge 2q-2$, against the guarantee
$n+1$. The sampler itself is not new: it is the fractional-revival decomposition of Kerr optics
\cite{bb1968, ys1986, mtk1990, tac1993, wg1996}, used for Kerr gates in \cite[Lemma~4]{uqc2026}. What is
new is that it carries over to the penalty layer, where it is the optimal decomposition and the unique
one. In the optimization reading one may choose the penalty angle on a rung, for instance by absorbing
it into the penalty weight, and pay $q$ at every register size; whether that restriction costs solution
quality is untested. A decomposition need only reproduce the penalty's moments on the levels where the
state entering the layer has weight: by the free recentring of Lemma~\ref{lem:free}, matching any
$2w+1$ consecutive levels costs exactly $\Gmin(\gamma, 2w)$, wherever the window sits, and against such
an approximate decomposition snapping to a rung buys an exact unitary and a closed-form sampler, not a
better scaling. For permutation-invariant unitaries that are not functions of $K$, such as the XY mixer expanded in \cite[App.~C]{anchor} at cost $O(n^3)$, the rational structure has no counterpart.

\subsection{Kerr branch expansions in bosonic simulation}\label{sec:pricing_kerr}

Two recent classical simulators of bosonic circuits \cite{uqc2026, gha2026} handle a Kerr gate
$e^{-i\gamma\hat n^2}$ by expanding it into phase-rotated copies, which act on a coherent state by
rotating its amplitude, and the number and weight of those branches is this paper's cost, gate by gate.
Theorem~\ref{main:C}(iv) settles the exact case; the same accounting applies to the Bose--Hubbard on-site term $e^{-iUt\,\hat n(\hat n-1)/2}$ ($\gamma = Ut/2$) and to the charging term of a superconducting island, whose linear parts are free by Lemma~\ref{lem:free} \cite[\S6.1]{v1}. At rational parameters Upreti, Quesada and Chabaud use
the $q$-term fractional-revival decomposition, which they take from \cite{tac1993, wg1996}, and bound
its $\ell_1$ norm by $\sqrt q$ \cite[Lemma~4, Thm.~16]{uqc2026}; among decompositions into phase
shifters that decomposition is optimal, so the bound is tight, and it is the unique optimum. At
irrational parameters no finite-cost decomposition over phase shifters exists at all, so the
approximation such simulators make is unavoidable; near low-denominator rationals a rational
approximation of the parameter is essentially optimal (Theorem~\ref{main:D}), and at badly approximable
parameters the truncated cost is $\Theta(n)$ whatever one does. How to approximate is then the windowed
problem above, on photon numbers: matching the Kerr phases only on the photon-number window of the
incoming state, recentred by Lemma~\ref{lem:free}, costs at most the window length, which grows like
$\sqrt{\bar n\log(1/\eta)}$ for mean photon number $\bar n$ and state error $\eta = \|\tilde\psi - e^{-i\gamma\hat n^2}\psi\|$, the norm distance between the approximate and the exact output state. Recentring the
construction of Guseynov, Holmes and Angrisani \cite[eqs.~(11)--(16)]{gha2026} on the state's window in
this way cuts the branches per Kerr gate by a factor growing like $\sqrt{\bar n}$, more than $10\times$ at
$\bar n = 10^4$ \cite[App.~K]{v1}, and the saving compounds over the Kerr layers of a circuit. Two
boundaries: with coherent or Gaussian inputs and linear optics every branch is classically samplable, so
this prices classical simulation, not a quantum resource; and the statement concerns decompositions
over phase shifters.

\subsection{One-axis twisting, read right to left}\label{sec:oat}

One-axis twisting is the workhorse nonlinearity of spin metrology: an ensemble of $n$ two-level atoms
prepared in a coherent spin state and evolved under $H = \chi J_z^2$ squeezes its spin noise at short
times \cite{ku1993}, and at longer times fractures into a superposition of a few coherent states, the
multicomponent cats of Bia\l{}ynicka-Birula \cite{bb1968}, Yurke and Stoler \cite{ys1986} and
Miranowicz, Tana\'s and Kielich \cite{mtk1990}, with the number of components set by the arithmetic of
the elapsed phase. The gate at time $t$ \emph{is} the chirp with $\gamma = \chi t$: on weight-$k$ states
$J_z$ has eigenvalue $n/2 - k$, so $\chi t\,J_z^2 = \chi t\,k^2$ plus linear and constant terms that are
free by Lemma~\ref{lem:free}, and global $z$-rotations are the resource set $S$. Here the physics fixes
the angle, the interaction implements the unitary directly, and the law says at which couplings the
cost is small --- those are the couplings at which cats form. A coherent spin state has weight on every
level, so by Proposition~\ref{prop:domination} the cost of the unitary is also the exact sampling
constant of the state it produces: the minimal signed weight of a rotation decomposition is a measure
of the twisted state. Three claims follow, each a corollary of \S\S\ref{sec:resonance}--\ref{sec:enlarge} under this change of names. (1)~At rational coupling $\chi t = \pi p/q$ the $q$-component decomposition of Theorem~\ref{main:B}, applied to a coherent state, is the multicomponent cat of \cite{bb1968, ys1986, mtk1990} (at $q = 2$ the Yurke--Stoler cat of \S\ref{sec:resonance}, up to the sign convention of $\chi$), and what is new is the converse: for $n \ge 2q-2$ it is \TV-optimal among all complex measures over rotations and the unique optimum; for smaller registers it is strictly beatable, so at small excitation cutoff the revival structure is \emph{not} the optimal representation. (2)~At badly approximable $\chi t/\pi$, Theorem~\ref{main:C} upgrades ``no revivals'' to: no cheap signed-mixture representation
over rotations exists at all, and by Theorem~\ref{main:E} this holds over arbitrary independent
single-spin $z$-rotations, with every quasi-probability decomposition over one ancilla-controlled layer
costing $\Theta(n)$ (Proposition~\ref{prop:channel}). (3)~The Kerr phase of a bosonic mode has no exactness threshold, since an optical coherent state
has weight on every photon number and Theorem~\ref{main:C}(iv) applies at every amplitude. This prices decomposition cost, not metrological performance: nothing here claims an improvement in squeezing, phase sensitivity or preparation fidelity. Thus the rational revival couplings at which finite-component cats form are precisely the resonant low-cost points of the rotation-decomposition problem, and the continued fraction of $\chi t/\pi$ decides which couplings those are \cite[\S10, App.~L]{v1}.

\section{Discussion}\label{sec:discussion}

\paragraph{The exchange rate, answered.} Section~\ref{sec:intro} asked how many shots buy back the $\Theta(n^2)$ entangling gates of the quadratic phase. The answer is the arithmetic of the angle, case by case in Table~\ref{tab:classification}. Under the ancilla-free guarantee these are the shot factors; under unbiased quasi-probability estimation they are squared.

\begin{table}[ht]\centering
\caption{The classification at a glance; $c_A$ and $q^*$ are the constants of Theorems~\ref{main:C} and~\ref{main:D}. Every entry is proved except those marked computed or measured.}
\label{tab:classification}
{\footnotesize\setlength{\tabcolsep}{5pt} \begin{tabular}{@{}p{3.6cm}p{5.3cm}p{3.7cm}l@{}}
\toprule
Angle $\gamma/\pi$ & Cost $\Gmin(\gamma, n)$ & Optimizer & Where \\
\midrule
rational $p/q$, $n \ge 2q-2$ & exactly $q$ & uniform $q$-point sampler, unique & Thm~\ref{main:B} \\
rational $p/q$, $n < n_0$, where $q \le n_0 \le 2q-2$ & $< q$, nondecreasing to $q$ at $n_0$ & computed;
$n_0 \approx 3q/2$ for the chirp & Section~\ref{sec:resonance} \\
badly approximable & $\Theta(n)$: $c_A(n+2) \le \Gmin \le n+1$ & DFT measure, optimal up to a constant
& Thm~\ref{main:C} \\
unbounded partial quotients & unbounded, but $\liminf_n \Gmin/n = 0$ & follows the convergents & Thm~\ref{main:C} \\
almost every angle & $n$ up to logarithmic factors, with dips along infinitely many scales & DFT measure, up to those factors & Thm~\ref{main:C}(v) \\
within $1/n^2$ of a rung with $q^* \le (n+2)/2$ & $\Theta(q^*)$, constants near $1$ & --- & Thm~\ref{main:D} \\
typical angle at fixed $n$ & $\ge 0.12(n+1)$ proved; $\approx 0.6(n+1)$ computed & --- & Prop.~\ref{prop:typical}; \S\ref{sec:pricing_lcu} \\
register squaring phase on $m$ qubits, any single-qubit gates & $2^{cm}$, $0.557 \le c \le 1$;
$c \approx 0.84$ measured & --- & Thm~\ref{main:register} \\
\bottomrule
\end{tabular}}
\end{table}

\paragraph{Scope.} For the fixed resource family $S = \{e^{i\theta K}\}$, the framework prices every diagonal unitary that is a function of $K$, and only those: any
other diagonal unitary admits no LCU over $S$ at any cost. The arithmetic needs quadratic $f$ and the
commutative case. The same duality prices other profiles: a single marked level, in the Hadamard frame
Grover's diffusion operator over collective $X$-rotations, costs below $9$ \cite[Prop.~G.2]{v1}; higher-degree
chirps keep the periodicity and lose the flatness \cite[App.~J]{v1}. For symmetric targets the lower bounds hold over one layer of arbitrary single-qubit gates (Theorem~\ref{main:E}); ancillas, feedforward and non-diagonal layers remain
open, and sampling overhead and tensor-network cost are not proxies for one another \cite{ggdp2026}.

\needspace{6\baselineskip}\paragraph{Open problems.}
\begin{enumerate}[leftmargin=2em, nosep]
\item Does the master formula hold at every angle, on the shoulders where it predicts linear growth with
detuning? \cite[Open Problem~M.1]{v1}
\item What do ancillas, feedforward and non-diagonal layers buy? Depth alone does not: two sampled rotation layers on either side of a fixed unitary that couples every level to every other cost exactly the product, and from three layers on the chirp arithmetic is gone; coherent control does, since two ancilla-controlled layers implement the $q = 2$ rung at cost one (\S\ref{sec:setup}). \cite[Open Problem~M.5]{v1}
\item From what register size is the $q$-angle sampler exactly optimal? The threshold lies in $[q, 2q-2]$,
computed near $3q/2$ for the chirp; what decides it? \cite[Open Problem~M.2]{v1}
\item The squaring phase on $m$ qubits costs $2^{cm}$ over single-qubit gates, $0.557 \le c \le 1$ proved and
$c \approx 0.84$ measured. What is $c$? \cite[Open Problem~M.3]{v1}
\item Does the distribution of $\Gmin/(n+1)$ at a uniformly random angle, with median about $0.6$ at every
$n$ computed, converge as $n \to \infty$? \cite[Open Problem~M.4]{v1}
\end{enumerate}

\section*{Acknowledgments}
We thank S.~Woerner and J.~Jakubik for helpful discussions.
This work was supported by the National Aeronautics and Space Administration under Award
No.~80MSFC25M0084. Z.~Wang is supported by NASA funding through cooperative agreement
80NSSC24M0035.

\begin{appendices}
\section{Proofs for Section~\ref{sec:setup}: the cost is a norm}\label{app:operational}

This appendix proves Propositions~\ref{prop:domination} and~\ref{prop:channel}.

\begin{proof}[Proof of Proposition~\ref{prop:domination}]
Ranging over $W$ and $x$, the condition is
\[|\langle\varphi|U\psi\rangle|^2 \;\le\; \Gamma\int|\langle\varphi|V(\theta)\psi\rangle|^2\,\mathrm{d}\nu(\theta)\]
for every vector $\varphi$, that is,
\[
U\psi\psi^\dagger U^\dagger \;\preceq\; \Gamma\int V(\theta)\psi\psi^\dagger V(\theta)^\dagger\,\mathrm{d}\nu(\theta).
\]
The dictionary $\mathcal V = \{e^{i\alpha}V(\theta)\psi\}$ is compact and closed under phases, its
mixtures are exactly the convex hull of $\{V(\theta)\psi\psi^\dagger V(\theta)^\dagger\}$ (Carath\'eodory),
and $U\psi = \sum_k u_k\psi_k$ lies in its span; so by \cite[Thm.~10]{regula2018} the least $\Gamma$
is $\min\{\|c\|_1^2 : U\psi = \sum_j c_j V(\theta_j)\psi\}$, for every $\nu$, discrete or not. Since
$V(\theta)\psi = \sum_k e^{ik\theta}\psi_k$ with the $\psi_k = P_k\psi$ orthogonal and all nonzero,
$U\psi = \sum_j c_jV(\theta_j)\psi$ holds iff $u_k = \sum_j c_je^{ik\theta_j}$ for every $k$, which is
\eqref{eq:moments}; by Theorem~\ref{main:A} the minimum is $\Gmin(U)$.

In our coordinates both inequalities are two lines, which we record because the witness is the
certificate of Theorem~\ref{main:A}. ($\le$) Take an optimal measure $\sum_j c_j\delta_{\theta_j}$
and $\nu_j = |c_j|/\|c\|_1$; then $\langle\varphi|U\psi\rangle = \sum_j c_j\langle\varphi|V(\theta_j)\psi\rangle$ and
Cauchy--Schwarz gives $|\sum_jc_ja_j|^2 \le \|c\|_1\sum_j|c_j||a_j|^2 = \|c\|_1^2\sum_j\nu_j|a_j|^2$ ---
the guarantee of \cite{chakraborty} at $\Gamma = \Gmin$, for complex coefficients as well as real.
($\ge$) The rank-one witness in the proof of \cite[Thm.~10]{regula2018} is here
$\varphi = \sum_k \bar a_k\,\psi_k/\|\psi_k\|^2$ with $a$ an optimal certificate: $\langle\varphi|V(\theta)\psi\rangle
= \sum_k a_ke^{ik\theta}$ has modulus at most $1$, while $|\langle\varphi|U\psi\rangle| = |\sum_k a_ku_k| = \sqrt{\Gmin}$,
so $\Gamma \ge |\langle\varphi|U\psi\rangle|^2/\sup_\theta|\langle\varphi|V(\theta)\psi\rangle|^2 \ge \Gmin$ for every $\nu$.
\end{proof}

\noindent In resource-theoretic terms, $\Gmin(U)$ is
one plus the generalized robustness of $U\psi$ with respect to the rotation orbit of any
full-support $\psi$, and $\sqrt{\Gmin}$ is Regula's atomic gauge $\Gamma_{\mathcal V}$ of that orbit
\cite{regula2018}; the operator-level Definition~\ref{def:cost} is the input-independent form.

\paragraph{Elementary estimators.} An elementary estimator consists of an ancilla register
$\mathcal A$ of any finite dimension (including one) in any state $\tau$, one collective rotation whose angle
is coherently selected by $\mathcal A$, $W = \sum_k P_k\otimes D^k$ with $D$ a unitary on $\mathcal A$, a
measurement of $\mathcal A$, and a real weight of modulus at most one applied to the measured value of the
system observable; write $-I \le E \le I$ for the weighted POVM. Its action is the block multiplier
$P_k\rho P_l \mapsto \mathrm{Tr}[ED^k\tau D^{-l}]\,P_k\rho P_l$. For the $\alpha$-th estimator of a decomposition we write $m^{(\alpha)}_{kl}$ for this multiplier.

\begin{proof}[Proof of Proposition~\ref{prop:channel}]
For the lower bound, let $A$ be an optimal certificate of Theorem~\ref{main:A} and put
$A(D) := \sum_ka_kD^k$. Then
\[
\textstyle\sum_{k,l}a_k\bar a_l\,\mathrm{Tr}[ED^k\tau D^{-l}] \;=\; \mathrm{Tr}[E\,A(D)\tau A(D)^\dagger]
\]
has modulus at most $\|A(D)\|^2 \le \|A\|_\infty^2 \le 1$ by the spectral theorem, while the same pairing of the
multiplier $u_k\bar u_l$ of $U\rho U^\dagger$ gives $|\sum_ka_ku_k|^2 = \Gmin$. Pairing the identity
$u_k\bar u_l = \sum_\alpha q_\alpha m^{(\alpha)}_{kl}$ with $a\bar a$ and bounding each term in absolute
value gives $\Gmin \le \sum_\alpha|q_\alpha|$. Upper bound: PREP--SELECT--PREP$^\dagger$ on an optimal
$\mu = \sum_jc_j\delta_{\theta_j}$ with the accept flag as $E$ has multiplier $u_k\bar u_l/\|c\|_1^2$,
a single term of weight $\|c\|_1^2 = \Gmin$. Ancilla-free: the multiplier of a signed mixture of
conjugations depends on $k-l$ only, and $u_k\bar u_l = g(k-l)$ forces a linear phase.
\end{proof}

\section{Proofs for Section~\ref{sec:resonance}: rational angles}\label{app:proofs_resonance}

\begin{proof}[Proof of uniqueness in Theorem~\ref{main:B}(ii)]
Write $s_j := \sqrt q\,\overline{\tilde c_j}$, which is unimodular, so that by the proof of
Theorem~\ref{main:B}(ii)
\[
P(\theta) \;=\; q^{-1}\textstyle\sum_j s_j\,e^{i(q-1)(\theta - \theta^*_j)}F_q(\theta - \theta^*_j).
\]
Off the atom coset every $F_q(\theta - \theta^*_j)$ is strictly positive
($F_q(\varphi) = 0$ exactly for $\varphi \in (2\pi/q)\mathbb{Z}\setminus 2\pi\mathbb{Z}$,
and $\theta - \theta^*_j$ lies in $(2\pi/q)\mathbb{Z}$ only at the atoms), so
$|P(\theta)| = 1$ would force equality in the triangle inequality, i.e.\ the $q$ phases
$s_j e^{-i(q-1)\theta^*_j} = s_j e^{2\pi i j/q} e^{-i\pi\sigma(q-1)/q}$ would have to
coincide: $s_j = C e^{-2\pi i j/q}$. Then
$\tilde c_j = \overline{C} e^{2\pi i j/q}/\sqrt q$ and
$u_k = \sum_j \tilde c_j e^{ik\theta^*_j}
= (\overline{C}/\sqrt q)\, e^{i\pi\sigma k/q} \sum_j e^{2\pi i j(k+1)/q}$,
which vanishes at $k = 0$ for $q \ge 2$ --- contradicting $|u_0| = 1$. Hence $|P| < 1$ off
the atoms. If $\nu$ is another minimizer then
$\int P \dd\nu = \sum_k a_k u_k = \sqrt q = \|\nu\|_{\TV}$ forces $\nu$ to be supported on
$\{|P| = 1\} = \{\theta^*_j\}$; a measure on $q$ known atoms is determined by its first $q$
moments, so $\nu = \mu^*$.
\end{proof}

\section{\texorpdfstring{Proofs for Section~\ref{sec:generic}: the dichotomy and the \mbox{master} formula}{Proofs for Section 4: the dichotomy and the master formula}}\label{app:proofs_master}

\begin{proof}[Proof of Corollary~\ref{cor:detunedUB}]
By \eqref{eq:submult} it suffices to bound $\Gmin(\delta, n)$. The space
$A(\mathbb{Z}_N)$ of functions on $\mathbb{Z}_N$ with norm
$\|u\|_A = \sum_j |\hat u(j)|$ is a Banach algebra under pointwise multiplication, and
the $N$-atom discrete-Fourier measure is feasible, so
$\sqrt{\Gmin(\delta, n)} \le \|e^{-i\delta f}\|_A$ for \emph{any} $f$ on $\mathbb{Z}_N$
agreeing with $k^2$ on $k = 0,\dots,n$ up to an affine term --- Lemma~\ref{lem:free} makes both
the constant and the linear phase free, so the constraint center may be chosen. Take the
\emph{centered} parabola $f(k) = (k - n/2)^2$, whose periodization is continuous. Removing the
mean and summing the exponential series, $\sqrt{\Gmin(\delta,n)} \le
\exp(|\delta|\,\|f\|_{A'})$ with $\|f\|_{A'} = \sum_{j \ne 0}|\hat f(j)|$. The second difference
of the centered $f$ on $\mathbb{Z}_N$ is $2$ except at $k = 0, n$, so
$\hat f(j) = (1 + e^{2\pi ij/N})/(4\sin^2(\pi j/N))$ and
$|\hat f(j)| = |\cos(\pi j/N)|/(2\sin^2(\pi j/N)) \le \tfrac12\csc^2(\pi j/N)$; with
$\sum_{j=1}^{N-1}\csc^2(\pi j/N) = (N^2-1)/3$,
\[
\|f\|_{A'} \;\le\; \frac{N^2 - 1}{6}. \qedhere
\]
\end{proof}

\begin{lemma}[Fej\'er-weighted detuning factor]\label{lem:fejerdetune}
For an integer $q \ge 1$ and real $\delta$ set
\[
\Psi_q(\delta) \;:=\; \frac1q\Big|\sum_{|m| < q}\Big(1 - \frac{|m|}{q}\Big)e^{-i\delta m^2}\Big|.
\]
Then $\Psi_q(0) = 1$,
$\Psi_q(\delta) \ge 1 - \delta^2\big(q^4/30 - q^2/12 + 1/20\big) \ge 1 - \delta^2 q^4/30$, and, centring the phase,
\[
\Psi_q(\delta) \;\ge\; 1 - \frac{\delta^2}{2}\Big(\frac{7q^4}{180} - \frac{q^2}{9} + \frac{13}{180}\Big) \;\ge\; 1 - \frac{7\,\delta^2q^4}{360}.
\]
\end{lemma}

\begin{proof}
$|\Sigma| \ge \operatorname{Re}\Sigma = \sum_m w_m\cos(\delta m^2)$ with $w_m = 1-|m|/q \ge 0$
and $\sum_m w_m = q$. By $\cos x \ge 1 - x^2/2$, $\operatorname{Re}\Sigma \ge q - (\delta^2/2)T(q)$
with $T(q) = \sum_{|m|<q}(1-|m|/q)m^4$. Summing $\sum_{m<q} m^4$ and $\sum_{m<q} m^5$ in closed
form gives the exact identity $T(q) = q^5/15 - q^3/6 + q/10$ for every $q \ge 1$; divide by $q$,
and use $q^3/6 \ge q/10$ for the last form. For the centred form, rotate before taking the real part:
$|\Sigma| \ge \operatorname{Re}(e^{i\delta\mu}\Sigma) = \sum_m w_m\cos(\delta(m^2-\mu))$ with
$\mu = (q^2-1)/6$, the $w$-mean of $m^2$, since $\sum_m w_m m^2 = (q^3-q)/6$. The same bound on the cosine
now pays the variance $\frac1q\sum_m w_m(m^2-\mu)^2 = T(q)/q - \mu^2 = 7q^4/180 - q^2/9 + 13/180$ instead of
the fourth moment, and $13/180 - q^2/9 < 0$.
\end{proof}

\begin{proof}[Proof of Theorem~\ref{thm:detunedLB}]
Write $u_k = U_k b_k$ with $U_k = e^{-i\pi(p/q)k^2}$ the rung data and $b_k = e^{-i\delta k^2}$.
Take the certificate of Theorem~\ref{main:B}(ii) for the rung, rotated and rephased by the free
action of Lemma~\ref{lem:free}:
\[
a_k = q^{-1/2}w_k\,\overline{U_k}\,e^{-i(\alpha k + \beta)},\quad w_k = 1 - \tfrac{|k-(q-1)|}{q},
\quad 0 \le k \le 2q-2,
\]
zero otherwise, with $\alpha = -2\delta(q-1)$ and $\beta = \delta(q-1)^2$. Admissibility is
untouched: $p(\theta) = e^{-i\beta}P(\theta - \alpha)$ for $P$ the polynomial of
Theorem~\ref{main:B}(ii), so $\|p\|_\infty = \|P\|_\infty \le 1$; the support lies in
$\{0,\dots,2q-2\} \subseteq \{0,\dots,n\}$, which is the only use of $n \ge 2q-2$. For the pairing,
$\overline{U_k}U_k = 1$, and putting $k = (q-1)+m$ the phase collapses,
$-\alpha k - \beta - \delta k^2 = -\delta m^2$, so
$\sum_k a_k u_k = q^{-1/2}\sum_{|m|<q}(1-|m|/q)e^{-i\delta m^2}$, of modulus
$\sqrt q\,\Psi_q(\delta)$. Theorem~\ref{main:A} gives the claim, and
Lemma~\ref{lem:fejerdetune} the explicit form.
\end{proof}

\label{app:proofs_generic}

\begin{proof}[Proof of Theorem~\ref{thm:convfloor}]
$q_k \le (n+2)/2$ is $n \ge 2q_k - 2$, the hypothesis of Theorem~\ref{thm:detunedLB} at the rung
$p_k/q_k$ (convergents are reduced). That theorem and Lemma~\ref{lem:fejerdetune} give
$\sqrt{\Gmin} \ge \sqrt{q_k}\,\Psi_{q_k}(\delta_k) \ge \sqrt{q_k}(1 - 7\delta_k^2 q_k^4/360)$, and
$7\delta_k^2 q_k^4/360 = \beta t_k^2$. The standard bound $|x - p_k/q_k| < 1/(q_kq_{k+1})$ with
$q_{k+1} \ge a_{k+1}q_k$ gives $t_k < 1/a_{k+1} \le 1$, and $1 - \beta = 0.808 > 0$. For the size of
$q_k$: $q_{k+1} > (n+2)/2$ by maximality of $k$, and $q_{k+1} \le (a_{k+1}+1)q_k$.
\end{proof}

\begin{proof}[Proof of Theorem~\ref{main:C}(i), details]
``$\Leftarrow$'' is Theorem~\ref{main:C}(ii). For ``$\Rightarrow$'' the tool is
Corollary~\ref{cor:detunedUB} at the rung $p_k/q_k$, which carries no condition on $n$: at $n = n_k$,
$N \le \sqrt{q_kq_{k+1}}$ gives $(N^2-1)/3 < q_kq_{k+1}/3$, and $|\delta_k| < \pi/(q_kq_{k+1})$, so
$\Gmin(\gamma, n_k) \le e^{\pi/3}q_k \le e^{\pi/3}(n_k+2)/\sqrt{a_{k+1}}$, using
$n_k + 2 \ge \sqrt{q_kq_{k+1}} \ge q_k\sqrt{a_{k+1}}$. Along a subsequence with $a_{k+1} \to \infty$ the
ratio tends to $0$. (Rational $x$ has $\Gmin = q$ bounded and is not badly approximable in this
convention.)
\end{proof}

\begin{proof}[Proof of Theorem~\ref{main:C}(iii)]
For rational $\gamma/\pi$, Theorem~\ref{main:B}(i) gives $\Gmin \le q$ for every $n$ and
Theorem~\ref{main:B}(ii) gives equality for $n \ge 2q-2$. If $\gamma/\pi$ is irrational its
continued-fraction expansion is infinite, so the convergent denominators $q_k$ increase without
bound. For $n \ge 2$ let $k(n)$ be the largest index with $q_{k(n)} \le (n+2)/2$; then
Theorem~\ref{thm:convfloor} gives $\Gmin(\gamma, n) \ge 0.6530\,q_{k(n)}$. As $n \to \infty$,
$k(n) \to \infty$ and hence $q_{k(n)} \to \infty$, and with it the lower bound. \qedhere
\end{proof}

\begin{proof}[Proof of Theorem~\ref{main:C}(iv)]
The constraint sets shrink as $n$ grows, so $\Gmin(\gamma, n)$ is nondecreasing and at most
$\Gamma_\infty(\gamma)$. If $L := \sup_n\Gmin(\gamma, n)$ is finite, pick $\mu_n$ feasible at level $n$
with $\|\mu_n\|_{\TV} \le \sqrt L + 1/n$. Bounded sets of measures on the circle are weak-$*$
sequentially compact, since $C(\mathbb{T})$ is separable, so a subsequence converges weak-$*$ to some
$\mu$. Then $\hat\mu(k) = \lim_j\hat\mu_{n_j}(k) = e^{-i\gamma k^2}$ for every $k$, and
$\|\mu\|_{\TV} \le \sqrt L$ by lower semicontinuity of the norm; hence $\Gamma_\infty = L$.
Theorem~\ref{main:C}(iii) makes $L$ equal to $q$ at rational angles and infinite otherwise. A measure
of norm $\sqrt q$ matching every moment is optimal for the problem at $n = 2q-2$, so it is $\mu^*$ by
Theorem~\ref{main:B}(ii).
\end{proof}

\begin{proof}[Proof of Theorem~\ref{main:C}(v)]
By Borel--Bernstein, for a.e.\ $x$ one has $a_j \le j(\log j)^{1+\varepsilon}$ for $j \ge j_0(x)$;
since $q_j \ge \varphi^{j-1}$, the index $k(n)$ is $O(\log n)$, so $a^* = O(\log n(\log\log n)^{1+\varepsilon})$
and Theorem~\ref{thm:convfloor} gives the lower bound. For the upper, Borel--Bernstein with the
non-summable $\varphi(j) = j\log j$ gives $a_{k+1} \ge k\log k$ infinitely often; by L\'evy,
$\log n_k = kL(1+o(1))$ with $L = \pi^2/(12\log 2)$, so $a_{k+1} \ge (1-o(1))(\log n_k/L)\log\log n_k$,
and Theorem~\ref{main:C}(i)'s rate gives the claim, since $e^{\pi/3}\sqrt L < 3.11$.
\end{proof}

\begin{proof}[Proof of Theorem~\ref{main:D}]
The hypothesis $q^* \le (n+2)/2$ is $n \ge 2q^*-2$, so Theorem~\ref{thm:detunedLB} applies at the
window rung and Corollary~\ref{cor:detunedUB} bounds above; insert $|\delta| \le 1/n^2$ into
$\Psi_{q^*}$ and into $\exp(|\delta|(N^2-1)/3)$. For the form in $F$: the window rung's term is at most $2q^*$,
every rung with $q > q^*$ has a term of at least $q > q^*$, and a rung with $q < q^*$ lies outside the window, so
it is at distance at least $\pi/(qq^*) - 1/n^2$ from $\gamma$ and its term is at least $\pi n^2/q^* \ge n+1$;
hence $q^* \le F \le 2q^*$.
\end{proof}

\begin{proof}[Proof of Proposition~\ref{prop:typical}]
Put $M = (n+2)/2$, $Q = \lfloor M\rfloor$ and $\Phi(u) = u(1-\beta u^2)^2$, which is increasing on $[0,1]$ since
$5\beta < 1$. Discard the rational $x$, a null set.

\emph{Cells.} For $q \le Q$ let $E_q$ be the set of $x$ whose top convergent (the $p_k/q_k$ of
Theorem~\ref{thm:convfloor}) has denominator $q$. On the side of $p/q$ where the parent fraction has
denominator $q'$, the $x$ having $p/q$ as a convergent are $x = p/q \pm 1/(q(qy + q'))$ with $y > 1$, and the
next convergent denominator is $\lfloor y\rfloor q + q'$ \cite{khinchin}. So $p/q$ is the top convergent exactly on
an interval of half-width $1/(qm)$, with $m$ the integer in $(Q, Q+q]$ congruent to $q'$ modulo $q$, and there
$t := q^2|x - p/q| < q/m < q/M$. As $p$ runs over the reduced residues modulo $q$, so do both parents' $q'$,
hence
\[
|E_q| \;=\; \frac2q\sum_{\substack{Q<m\le Q+q\\ \gcd(m,q)=1}}\frac1m \;<\; \frac{2\varphi(q)}{qM} \;\le\; \frac2M .
\]

\emph{The upper bound on the share.} On $E_q$, Theorem~\ref{thm:convfloor} with the centred bound of
Lemma~\ref{lem:fejerdetune} gives $\Gmin \ge q(1 - \beta t^2)^2 \ge M\Phi(q/M)$. If $\Gmin < c(n+1)$ then
$\Phi(q/M) < c(n+1)/M$, so $q < uM$ with $\Phi(u) = c(n+1)/M$ (the bound is trivial if $c(n+1)/M \ge \Phi(1)$),
and the set has measure below $(2/M)\,uM = 2u \le 2c(n+1)/(M(1-\beta)^2) < 6.13\,c$. For the median take
$c = \Phi(1/4)/2 = 0.122$: then $c(n+1)/M < 2c = \Phi(1/4)$, so $u < 1/4$ and the measure is below $1/2$. For
$6 \le n \le 6000$ the bound $0.19$ sums the exact $|E_q|$ over the $q$ with $\Phi(q/M) < c(n+1)/M$, in
floating point with a margin of $10^{-9}$. For $n > 6000$, $|E_q| < 2\varphi(q)/(qM)$ and
$\sum_{q \le X}\varphi(q)/q \le 6X/\pi^2 + \ln X + 7/3$ (write $\varphi(q)/q = \sum_{d \mid q}\mu(d)/d$, and use
$\sum_{d > X}d^{-2} < 1/(\lfloor X\rfloor + \frac12)$ and $\sum_{d \le X}1/d \le 1 + \ln X$) bound the measure by
$12u/\pi^2 + 4(\ln(u(n+2)/2) + 7/3)/(n+2)$. At $u = 0.406$ this is below $1/2$ for $n = 6000$, and it decreases
in $n$, so the median exceeds $\Phi(0.406)/2 > 0.19$.

\emph{The lower bound on the share.} Put $X = (n+1)/f$ and $N = n+1$. For $x$ within $r_q = 3\ln(X/q)/(\pi(N^2-1))$ of a reduced
$p/q$ with $q \le X$, Corollary~\ref{cor:detunedUB} with $\delta = \pi(x - p/q)$ gives $\Gmin \le q\,e^{\ln(X/q)} = X$.
Distinct such fractions are at distance at least $1/(qq')$, and
$(r_q + r_{q'})qq' \le 3X^2(2/e)/(\pi(N^2-1)) \le 0.94/f^2 < 1$ by $y\ln(1/y) \le 1/e$, so the windows are
disjoint, of total measure
\[
\frac{6}{\pi(N^2-1)}\sum_{q\le X}\varphi(q)\ln\frac Xq \;\ge\; \frac{6}{\pi f^2}\,\inf_{X \ge 2}\,\frac{1}{X^2}\sum_{q\le X}\varphi(q)\ln\frac Xq .
\]
The infimum is at least $0.15196$: it was computed over $2 \le X \le 10^6$ (on each $[m, m+1)$ the function has at
most one critical point, a maximum, so its infimum there is attained at an integer), and beyond $10^6$ it
follows from $\sum_{q \le t}\varphi(q) \ge 3t^2/\pi^2 - 2t - t\ln t$, which gives the bound
$3/(2\pi^2) - (\ln X + 1)/X$. Hence the measure is at least $0.290/f^2$.
\end{proof}

\section{Proofs for Section~\ref{sec:enlarge}: the enlarged dictionaries}\label{app:proofs_enlarge}

\begin{proof}[Proof of Theorem~\ref{main:E}]
$\le$ is inclusion of resource sets. For $\ge$: duality over $C(\mathbb{T}^n)$ reads, as
in Theorem~\ref{main:A},
$\sqrt{\Gamma_{\mathbb{T}}} = \max |\sum_x a_x u_{|x|}|$ over multiaffine
polynomials $A(z) = \sum_x a_x \prod_{j \in x} z_j$ with $\|A\|_{\mathbb{T}^n} \le 1$
(the constrained moments are exactly the frequencies $x \in \{0,1\}^n$). Averaging $a$
over coordinate permutations fixes the pairing (the data are symmetric) and does not
increase the sup-norm, so we may take $a_x = \alpha_{|x|}$, i.e.\
$A = \sum_k \alpha_k e_k(z)$ with $e_k$ the elementary symmetric polynomials. $A$ is
symmetric and multiaffine, and the closed unit disc is a convex circular region, so the
Grace--Walsh--Szeg\H{o} coincidence theorem \cite{rahmanschmeisser} applies: for every
$z \in \mathbb{T}^n \subset \overline{\mathbb{D}}^n$ there is a single
$w \in \overline{\mathbb{D}}$ with $A(z) = A(w, \ldots, w) = p(w)$, where
$p(w) = \sum_k \alpha_k \binom{n}{k} w^k$. Hence
$\|A\|_{\mathbb{T}^n} \le \sup_{\overline{\mathbb{D}}} |p| =
\|p\|_{\partial\mathbb{D}}$ by the maximum principle, with $\ge$ from restriction to the
diagonal: the two norms are \emph{equal}. The pairing is
$\sum_k \binom{n}{k} \alpha_k u_k = \sum_k b_k u_k$ with $b_k$ the
coefficients of $p$ --- so the torus dual is the circle dual of
Theorem~\ref{main:A}, variable for variable. Both values are attained; they coincide.
\end{proof}

\begin{proof}[Proof of Theorem~\ref{main:F}]
Using sheet $c$ alone requires $\hat\mu_c(k)\,e^{ick^2} = e^{-i\gamma k^2}$, so the cost of sheet $c$ is $\Gmin(\gamma + c)$ by definition, the sheet's phase factoring
out of every atom, giving $\Gamma^{D(\ell)} \le \min_c \Gmin(\gamma+c)$. Write $q = 2^a q'$ with $q'$ odd.
The sheet angles are $(\gamma + c)/\pi = p/q + j/2^\ell$ with denominator $Q_j$. If $a \le \ell$,
then $p/q + j/2^\ell = (p\,2^{\ell-a} + jq')/(2^\ell q')$, and since $q'$ is odd the choice
$j \equiv -p\,2^{\ell-a}(q')^{-1} \bmod 2^\ell$ makes $2^\ell$ divide the numerator, while
$\gcd(\cdot, q') = \gcd(p, q') = 1$; hence $\min_j Q_j = q'= q/2^{v_2(q)}$. If $a > \ell$ then
$p$ is odd and $j\,2^{a-\ell}q'$ is even, so the numerator $p + j\,2^{a-\ell}q'$ is odd and
coprime to $q'$, giving $Q_j = q$ for every $j$. Theorem~\ref{main:B}(i) then supplies
the value on the best sheet. For the converse the dual of the enlarged dictionary requires one coefficient vector $a$
with $\|p_c\|_\infty \le 1$ on \emph{every} sheet, $p_c(\theta) = \sum_k a_k e^{ick^2}
e^{ik\theta}$, and pairing $\sum_k a_k u_k = \sqrt{Q_{\min}}$. Two cases.

If $v_2(q) > \ell$, every sheet rate is again a rung of denominator exactly $q$: writing
$q = 2^a q'$ with $q'$ odd and $a > \ell$, the rate $\gamma + c = \pi(p + j2^{a-\ell}q')/(2^aq')$
has odd numerator, coprime to $q'$. So for $a_k = q^{-1/2}w_k\overline{u_k}$ the Fej\'er
certificate of Theorem~\ref{main:B}(ii), $p_c$ is \emph{verbatim} that theorem's certificate
for the pair $(P_c, q)$ --- same denominator, hence same window length --- and
$\|p_c\|_\infty \le 1$ holds on all $2^\ell$ sheets at once with the pairing unchanged at
$\sqrt q$. Nothing needs lengthening.

If $v_2(q) \le \ell$, put $s = \ceil{\ell/2}$, $L = 2^s$, and \emph{thin} the certificate onto
$L\mathbb{Z}$. Since $c(Ll)^2 = \pi j 2^{2s-\ell}l^2$ is $\equiv 0$ for $\ell$ odd and
$\equiv \pi j l$ for $\ell$ even (using $l^2 \equiv l \bmod 2$), every sheet acts on the
sublattice by a \emph{free linear phase}, so a certificate supported on $L\mathbb{Z}$
satisfies $p_c(\theta) = \tilde P(L\theta + \varepsilon_c)$ for a single $\tilde P$ and
$2^\ell$ constants $\varepsilon_c$: all $2^\ell$ dual constraints collapse to one. Take
$\tilde P$ to be the certificate of Theorem~\ref{main:B}(ii) for the decimated rate
$2^{2s}\gamma = \pi P/Q_{\min}$ (the denominator is exactly $Q_{\min}$ because
$2s \ge \ell \ge v_2(q)$),
\[
a_{Ll} = \frac{1}{\sqrt{Q_{\min}}}\Bigl(1 - \frac{|l-(Q_{\min}-1)|}{Q_{\min}}\Bigr)
e^{i\pi P l^2/Q_{\min}}, \quad l = 0,\ldots,2Q_{\min}-2, \qquad a_k = 0 \text{ otherwise};
\]
then $\|p_c\|_\infty = \|\tilde P\|_\infty \le 1$ for every sheet and
$\sum_k a_k u_k = \sqrt{Q_{\min}}$, since $u_{Ll} = e^{-i(2^{2s}\gamma)l^2}$. The stated
degrees are the supports of the two certificates.
\end{proof}

\section{Proofs for Section~\ref{sec:separation}: the separation}\label{app:proofs_separation}

\begin{proof}[Proof of Theorem~\ref{main:register}(a)]
\emph{Lower bound.} Cut the register into its $a$ low and $b = m-a$ high bits, $k = k_A + 2^ak_B$. Then
$k^2 = k_A^2 + 2^{a+1}k_Ak_B + 4^ak_B^2$, so $G = D_AFD_B$ with unitary diagonal $D_A, D_B$, which do not
change singular values, and $F[k_A, k_B] = e^{2\pi i k_Ak_B/2^{b-1}}$. $F$ depends on $k_A$ only modulo
$2^{b-1}$; when $a \ge b-1$ its rows run over whole periods and
$F = \mathbf{1}_{2^{a-b+1}}\otimes \mathrm{DFT}_{2^{b-1}}\otimes \mathbf{1}_2^{\top}$ (an all-ones column,
the unnormalised discrete Fourier matrix, an all-ones row), whose $2^{b-1}$ nonzero singular values all
equal $2^{(a+1)/2}$. So $\|F\|_{\mathrm{tr}} = 2^{b-1}\,2^{(a+1)/2}$, Lemma~\ref{lem:cut} gives
$\Gamma_{\mathbb{T}}(f_m) \ge 2^{b-1}$, and $b = \lfloor (m+1)/2\rfloor$ is the largest admissible choice.
\emph{Upper bound.} Modulo $2^m$, $k^2 = \sum_{2j<m}4^jx_j + \sum_{i<j,\ i+j\le m-2}2^{i+j+1}x_ix_j$: bit
$m-1$ enters only linearly, and the pairs $(0,j)$ contribute the controlled one-local phase
$e^{2\pi i x_0k'/2^{m-1}}$, $k' = \sum_{j\ge1}2^jx_j$, which equals $\tfrac12(1+Z_0) + \tfrac12(1-Z_0)V$
with $V$ one-local: four product terms of weight $\tfrac12$. The remaining pairs are those of $f_{m-2}$ on
bits $1, \ldots, m-2$, so $\sqrt{\Gamma_{\mathbb{T}}(f_m)} \le 2\sqrt{\Gamma_{\mathbb{T}}(f_{m-2})}$, and
$\sqrt{\Gamma_{\mathbb{T}}(f_2)} = 1$, $\sqrt{\Gamma_{\mathbb{T}}(f_3)} = \sqrt2$ seed both parities.
The one-parameter value $2^{m-1}$ is Theorem~\ref{main:B}(ii) at the rung $1/2^{m-1}$, since
$n = 2^m - 1 \ge 2q - 2$; the exact small cases are \cite[Rem.~H.2]{v1}.

\noindent The recursion starts from $\Gamma_{\mathbb{T}}(f_2) = 1$, since $k^2 \equiv x_0 \pmod 4$ makes $f_2 = i^{x_0}$ a one-local phase, and from $\Gamma_{\mathbb{T}}(f_3) \le 2$, since $k^2 \equiv x_0 + 4x_1 + 4x_0x_1 \pmod 8$ makes $f_3$ a $CZ$ up to one-local phases and $e^{i(\pi/4)Z_0Z_1} = (1 + iZ_0Z_1)/\sqrt2$. That no single cut certifies more than the lower bound follows from the inert top bit: $f_m(x)$ does not depend on $x_{m-1}$, so for every bipartition of the $m$ qubits the reshaped matrix $G_{A|B}$ is the $(m-1)$-qubit matrix with its rows or its columns duplicated, whose trace norm is $\sqrt2$ times that of the $(m-1)$-qubit matrix; against the normalization $2^{m/2} = \sqrt2\cdot 2^{(m-1)/2}$ of Lemma~\ref{lem:cut} this is exactly the $(m-1)$-qubit cut bound, which the ceiling stated after Lemma~\ref{lem:cut} caps at $2^{\lfloor (m-1)/2\rfloor}$, the value the low-bits/high-bits cut attains.
\end{proof}

\begin{proof}[Proof of Theorem~\ref{main:register}(b)]
\emph{The top bit is free.} For $m \ge 2$, $(k+2^{m-1})^2 \equiv k^2 \pmod{2^m}$, so $f_m$ does not
depend on $x_{m-1}$. Restricting a decomposition to $x_{m-1} = 0$, and extending one by the angle $0$
on that qubit, gives $\Gamma_{\mathbb{T}}(f_m) = \Gamma_{\mathbb{T}}(h_m)$ for the restriction $h_m$ of
$f_m$ to the $m-1$ low bits.

\emph{Splitting.} Let $m \ge 6$ and $M = m-4$. Write $k = k_L + 4c + 2^{m-3}y + 2^{m-1}x_{m-1}$, with
$k_L, y \in \{0, \ldots, 3\}$ carried by bits $0, 1$ and $m-3, m-2$, and $c$ by the $M-1$ bits in
between. With $u = k_L + 4c$ and $v = y + 4x_{m-1}$, $k^2 = u^2 + 2^{m-2}uv + 2^{2m-6}v^2$; the last term
vanishes modulo $2^m$ because $m \ge 6$, and $uv \equiv k_Ly \pmod 4$. Hence, modulo $1$,
\[
\frac{k^2}{2^m} \;\equiv\; \frac{k_L^2}{2^m} + \frac{k_Ly}{4} + \frac{2k_Lc}{2^{m-2}} + \frac{c^2}{2^{M}},
\qquad
h_m(x) = \omega^{k_L^2}\,\zeta_{k_L}^{\,x_Y}\,e^{i\sigma_{k_L}\cdot x_C}\,h_M(x_C),
\]
with $\omega = e^{2\pi i/2^m}$, $\zeta_k = (i^k, (-1)^k) \in \mathbb{T}^2$ read against the two bits of
$y$, and $\sigma_k \cdot x_C = 2\pi\cdot 2kc/2^{m-2}$ a one-local phase on the middle block. For fixed
$k_L$, the target is $h_M$ twisted by a one-local rotation and multiplied by a character of $y$.

\emph{Lifting.} Let $A$ be an optimal certificate for $h_M$: multiaffine in the $M-1$ middle variables
$z_C$, with $\|A\|_{\mathbb{T}} \le 1$ and pairing $\sqrt{\Gamma_{\mathbb{T}}(f_M)}$. With the two-bit
Fourier projectors $Q_k(w) = \frac14\prod_{t=0,1}\bigl(1 + e^{-2\pi i k2^t/4}w_t\bigr)$, which satisfy
$Q_k(\zeta_k) = 1$, put
\[
A^\sharp(z_L, z_C, z_Y) \;:=\; \sum_{k<4} z_L^{x(k)}\,\omega^{-k^2}\,Q_k(z_Y)\,A\bigl(e^{-i\sigma_k}\odot z_C\bigr),
\]
with $x(k)$ the bit vector of $k$ and $e^{-i\sigma_k}\odot z_C$ the coordinatewise twist. Each summand
is a product of multiaffine factors in disjoint variables, so $A^\sharp$ is multiaffine, and on the torus
$|A^\sharp| \le \sum_k |Q_k(z_Y)|$. Its coefficient at $(x_L, x_C, x_Y)$ is
$\omega^{-k_L^2}\,q_{k_L}(x_Y)\,e^{-i\sigma_{k_L}\cdot x_C}\,a(x_C)$, and against the splitting the phases
$\omega^{\pm k_L^2}$ and $e^{\pm i\sigma_{k_L}\cdot x_C}$ cancel:
$\langle A^\sharp, h_m\rangle = \sum_k Q_k(\zeta_k)\,\langle A, h_M\rangle = 4\sqrt{\Gamma_{\mathbb{T}}(f_M)}$.

\emph{The gadget.} On $w = (e^{i\theta_0}, e^{i\theta_1})$, with $\psi = \theta_0/2$ and
$|1+e^{i\phi}|/2 = |\cos\frac\phi2|$, one finds
$\sum_k |Q_k(w)| = |\cos\tfrac{\theta_1}{2}|\,X + |\sin\tfrac{\theta_1}{2}|\,Y$ with
$X = |\cos\psi| + |\sin\psi|$ and $Y = |\cos(\psi - \frac\pi4)| + |\sin(\psi - \frac\pi4)|$. The
supremum over $\theta_1$ is $(X^2+Y^2)^{1/2}$, and
$X^2 + Y^2 = 2 + |\sin2\psi| + |\cos2\psi| \le 2 + \sqrt2 = 4\cos^2\frac\pi8$. So
$A^\sharp/(2\cos\frac\pi8)$ is admissible, and
\[
\sqrt{\Gamma_{\mathbb{T}}(f_m)} \;\ge\; 2\sec\tfrac{\pi}{8}\,\sqrt{\Gamma_{\mathbb{T}}(f_{m-4})},
\qquad 4\sec^2\tfrac{\pi}{8} \;=\; 8(2-\sqrt2) \;=\; 4(4-2\sqrt2).
\]

\emph{Iteration.} Apply this $t = \lfloor (m-2)/4\rfloor$ times, down to $M' = m-4t \in \{2, \ldots, 5\}$,
and finish with the cut bound $\Gamma_{\mathbb{T}}(f_{M'}) \ge 2^{\lfloor (M'-1)/2\rfloor}$ of
Theorem~\ref{main:register}(a). Since $2t + \lfloor (M'-1)/2\rfloor = \lfloor (m-1)/2\rfloor$, this is the
stated bound.
\end{proof}

\begin{proof}[Proof of the ceiling in Remark~\ref{rem:liftcap}]
Each readout value is an integral, $Q_k(\zeta_k) = \int Q_k\,d\nu^{(k)}$, against a product of measures on
the circle that represent evaluation of affine functions: density $\frac\pi2 e^{-i(\theta-a)/2}$ on
$(a-\pi, a+\pi)$ for the point $e^{ia}$, and $\frac{1\mp i}{2}\delta_{i} + \frac{1\pm i}{2}\delta_{-i}$ for
$\pm1$. (These are dual extremal kernels for point evaluation on the span of $1$ and $z$, cf.\ Rogosinski
and Shapiro \cite{rogosinskishapiro1953}; the mass $\pi/2 = \|1+z\|_\infty/\|1+z\|_{L^1}$ is the least
possible for a representing measure dominated by Lebesgue measure, as testing on $1+z$ shows.) All $2^j$
of these products are dominated by one measure of mass $\sqrt2\,(\pi/2)^{j-1}$, so
$|\sum_kQ_k(\zeta_k)| \le \sqrt2\,(\pi/2)^{j-1}\sup\sum_k|Q_k|$.
\end{proof}

\end{appendices}

\begingroup\raggedright

\endgroup

\end{document}